\documentclass[11pt]{amsart}
\usepackage{geometry}                
\usepackage{graphicx}
\usepackage{enumerate}
\usepackage{subfig}
\usepackage{amssymb}
\usepackage{epstopdf}
\usepackage{mathrsfs,amssymb}
\usepackage{graphicx}
\usepackage{amscd}
\usepackage{amsmath}
\usepackage{amssymb}
\usepackage{amsthm}
\usepackage{amsfonts}
\usepackage{fancyhdr}
\usepackage{amsxtra}
\usepackage{tikz}
\usepackage{setspace}
\usepackage{amsrefs}
\newtheorem{theorem}{Theorem}
\newtheorem{lemma}{Lemma}
\newtheorem{corollary}{Corollary}
\newtheorem{proposition}{Proposition}

\newtheorem{claim}{Claim}

\newtheorem{remark}{Remark}

\newtheorem{definition}{Definition}

\def\bt{\begin{theorem}}
\def\et{\end{theorem}}
\def\bp{\begin{proposition}}
\def\ep{\end{proposition}}
\def\bl{\begin{lemma}}
\def\el{\end{lemma}}

\numberwithin{equation}{section}
\numberwithin{theorem}{section}
\numberwithin{proposition}{section}
\numberwithin{lemma}{section}
\numberwithin{corollary}{section}
\numberwithin{remark}{section}
\numberwithin{definition}{section}
\numberwithin{claim}{section}

\begin{document}

\title{on solutions of the 2D Navier-Stokes equations with constant energy and enstrophy}

\date{\today}
\author{J. Tian$^{1}$}
\address{$^1$Department of Mathematics\\
Texas A\&M University\\ College Statin, TX 77843}
\author{B. S. Zhang$^{1}$}
\address{$\dagger$ to whom correspondence should be made}
\email[J.Tian]{jtian@math.tamu.edu}
\email[B. S. Zhang$^\dagger$]{bszhang@math.tamu.edu}

\begin{abstract}
It is not yet known if the global attractor of the space periodic 2D Navier-Stokes equations contains nonstationary solutions $u(x,t)$ such that their energy and enstrophy per unit mass are constant for every $t \in (-\infty, \infty)$. The study of the properties of such solutions was initiated in 
\cite{CMM13}, where, due to the hypothetical existence of such solutions, they were called ``ghost solutions". In this work, we introduce and study geometric structures shared by all ghost solutions. This study led us to consider a subclass of ghost solutions for which those geometric structures have a supplementary stability property. 
In particular, we show that the wave vectors of the active modes of this subclass of ghost solutions must satisfy certain supplementary constraints. We also found a computational way to check for the existence of these ghost solutions.
\end{abstract}

\maketitle
\section{Introduction}
The universal features of turbulent 2D flows are determined by the time invariant probability measures carried by the global attractor $\mathcal{A}$ of the 2D Navier-Stokes equations(NSE) \cite{FJMR02}. Some of these features involve the time average of energy per unit mass $e(t)$, enstrophy per unit mass $E(t)$, and palinstrophy per unit mass $P(t)$. Therefore the study of the triplet $(e(t),E(t),P(t)), t\in \mathbb{R}$, associated with a solution $u(t)\in \mathcal{A}$, for all  $t\in \mathbb{R}$, is of intrinsic interest. 

A question which arises naturally in this context is whether the past (i.e., $t\leq 0$) of $(e(t),E(t),P(t))$ is sufficient to uniquely determine the solution $u(t)$ for all $t\in \mathbb{R}$. A particularly intriguing version of this problem was considered in \cite{CMM13}, 
namely, whether or not it is possible for $(e(t),E(t))$, corresponding to a nonstationary solution $u(t)$, with $u(t) \in \mathcal{A}$, for all $t\in \mathbb{R}$, to be constant in time. In this paper, we continue the study of such solutions, which were called ghost solutions in \cite{CMM13}. We assume, as in \cite{CMM13}, that the force $g$ is an eigenvector of the Stokes operator $A$, namely, $Ag=\lambda g$, where, $\lambda\geq 2 \kappa_0^2$, with $\kappa_0=2\pi/L$ and $L>0$ the period. In the case $\lambda=\kappa_0^2$, the global attractor $\mathcal{A}$ is reduced to a single point (see \cite{MARC87}), hence there are no ghost solutions.

One interesting feature of ghost solutions is that, with only one exception, the scalar products of any two different terms in the functional form of the NSE are constants in time. In order to extract geometric properties from this feature, we introduce two different reference frames that are rigidly carried by the solutions.

We first define a moving reference frame consisting of four orthonormal functions in the real separable Hilbert phase space $H$. Then we project each term of the NSE onto the reference frame. The Stokes operator $A$ and the inertial term in the NSE can be represented by a symmetric positive definite matrix and by a tensor, respectively. It turns out that with two possible exceptions, the coordinates of these above specified projections are constant in time. Moreover, we find there is an unitary transformation between any two ghost solutions. However,  if either a sequence of ghost solutions converge to a stationary solution, or a ghost solution converges to a fixed point when $t \rightarrow \infty$, this reference frame loses stability and degenerates to a three dimensional frame.

This kind of instability problem is avoided by introducing another reference frame. For the new reference frame, the degeneracy problem occurs only when there exist some $t \in \mathbb{R}$, such that $A^2u(t)=\gamma(t) g+\beta(t) u(t)+\alpha(t) Au(t)$. Denote the set of these $t$'s as $C$. We show that if the set $C$ has an accumulation point, then the coefficients $\gamma(t), \beta(t)$ and $\alpha(t)$ are all constants in time. For convenience, if a ghost solution satisfies the relation $A^2u(t)=\gamma g+\beta u(t)+\alpha Au(t)$, for all $t \in \mathbb{R}$, and for some time independent coefficients $\gamma$, $\beta$, and $\alpha$, we will call it chained.

It can be shown that any chained ghost solution $u(t)$ will be a sum of three components. One component is a scalar multiple of $g$, and the other two  $u_+$, $u_-$ are eigenvectors of $A$, namely,
 $Au_+=\mu_+u_+$, $Au_-=\mu_-u_-$ ($\mu_-<\lambda<\mu_+$). It turns out that numerically finding a possible chained ghost solution is equivalent to solving a Galerkin system of a particular form. 

In the particular case $\lambda=2 \kappa_0^2$, we can show that $\mu_-=\kappa_0^2$ and $\mu_+=5 \kappa_0^2$. And as a consequence, we prove that there is no chained ghost solutions. In the process of  proving this result, we show that the chained ghost solutions in this case, if they existed, must project only on a small portion of a parabolic curve in the $e,E$-plane. Then, a careful exploitation of the interaction between the wave numbers and active modes in this Galerkin system leads to several useful annihilations that allow us to show that there do not exist chained ghost solutions. This, in turn, guarantees the stability of the new reference frame. 

The paper is organized as follows. Section 2 gives the background and settings for our discussion. In Section 3, we present some inequalities which are necessary later and also recall some results from \cite{CMM13}. In Section 4, the geometric structure associated with a ghost solution is constructed, as is a unitary transformation between two ghost solutions. Section 5 shows that the closure of $GHO:=\{u(\cdot)\in \mathcal{A}:u(\cdot) \text{ is a ghost solution}\}$ in $C(\mathbb{R}, \mathcal{A})$ is the union of $GHO$ and a set of all stationary solutions. Section 6 describes the geometric behavior of chained ghost solutions in detail. In Section 7, a particular Galerkin method for finding a possible chained ghost is given. Section 8 has the proof that there do not exist chained ghost solutions in the case when $\lambda=2\kappa_0^2$.


\section{Preliminaries}
\subsection{Functional Analysis Framework}
In this paper, we consider the following imcompressible Navier-Stokes equations (NSE) in $\Omega=[0,L]^2$
\begin{align*}
 &\frac{\partial u}{\partial t}-\nu \bigtriangleup u+(u \cdot \bigtriangledown ) u+\bigtriangledown p=F,
\\ &div u=0,
\\ &u(x,t_0)=u_0 (x),
\\ &\int_{\Omega} u dx=0, \int_{\Omega} F dx=0,
\end{align*}
where $u(x,t)$ is the velocity of fluid at time $t$, at point $x$; $u$ and $p$ are unknown, $\Omega$-periodic functions, and $\nu>0$ is the kinematic viscosity of the fluid. $L>0$ is the period, $p$ is the pressure, and $F$ is the "body" force as in \cite{CF89}, \cite{Temam83}, \cite{Temam97}.
The real separable Hilbert space $\mathit{H}$ is formed by the set of all $\mathbb{R}^2$-valued functions $u(x), x\in \mathbb{R}^2$, which has the Fourier expansion
\begin{align*}
u(x)=\sum_{k\in \mathbb{Z}^2\setminus\{0\}}\hat{u}(k)e^{i\kappa_{0}k\cdot x}
\end{align*}\\
where $\kappa_{0}=2\pi/L$ and $\hat{u}(k)\in \mathbb{C}^2 $, for all $k\in \mathbb{Z}^2\setminus\{0\}$, satisfying
\begin{align}
\label{reality_cond}
\hat{u}(k)^{*}=\hat{u}(-k),
\end{align}
\begin{align}
\label{zero_div}
 k \cdot \hat{u}(k)=0,
\end{align}
 for all $ k\in\mathbb{Z}^2\setminus\{0\}$,
and
\begin{align*}
\sum_{k\in \mathbb{Z}^2\setminus\{0\}}\hat{u}(k) \cdot  \hat{u}(k)^{*}< \infty.
\end{align*}
The scalar product in $\mathit{H}$ is taken to be the usual $[L^{2}(\Omega)]^2$-product
\begin{align*}
 (u,v)=\int_{\Omega}u(x) \cdot v(x) dx.
\end{align*}
Denote by $A$ the Stokes operator,  which is the linear map defined by
\begin{align*}
\mathit{A}:  \mathit{D}(\mathit{A})\mapsto  \mathit{H}, \mathit{A}\hat{u}(k)={\kappa_{0}}^2(k\cdot k)\hat{u}(k), \forall k\in\mathbb{Z}^2\setminus\{0\},
\end{align*}
where
\begin{align*}
 u\in\mathit{D}(\mathit{A})=\{u\in \mathit{H}: \sum_{k\in \mathbb{Z}^2\setminus\{0\}}(k\cdot k)^2\hat{u}(k) \cdot  \hat{u}(k)^{*}< \infty\}.
\end{align*}
It is well known that the operator $A$ is self-adjoint, with a compact inverse $A^{-1}$, and the specturm (cf. \cite{CF89},\cite{Temam97})
\begin{align*}
sp(\mathit{A}) \subseteq  \{\lambda_0, 2\lambda_0,3\lambda_0...\}, \lambda_0={\kappa_{0}}^2 \in sp(\mathit{A}).
\end{align*}
The powers of $A$ are defined by
\begin{align*}
 \mathit{A}^{\alpha}\hat{u}(k)={\kappa_{0}}^{2\alpha}(k\cdot k)^{\alpha}\hat{u}(k), \forall k\in\mathbb{Z}^2\setminus\{0\},
\end{align*}
with domain given by
\begin{align*}
 \mathit{D}({\mathit{A}}^{\alpha})=\{u\in \mathit{H}: \sum_{k\in \mathbb{Z}^{2}\setminus\{0\}}(k\cdot k)^{2\alpha}\hat{u}(k) \cdot  \hat{u}(k)^{*}< \infty\}.
\end{align*}
$\mathit{D}({\mathit{A}}^{\alpha})$ is a Hilbert space with the inner product $(\mathit{A}^{\alpha/2}u,\mathit{A}^{\alpha/2}u)$. In the particular cases when $\alpha=0,\frac{1}{2}$, we use the notation
\begin{align*}
 |u|=(u,u)^{1/2} , ||u||=({\mathit{A}}^{1/2}u,{\mathit{A}}^{1/2}u)^{1/2},
\end{align*}
for the norms on $D(A^0)=H$, and $V:=D(A^{1/2})$, respectively.
Denote by $\mathit{B}$ the $\mathit{H}$-valued bilinear map defined on $ \mathit{D}({\mathit{A}})\times  \mathit{D}({\mathit{A}})$ by
\begin{align}
\label{coef_for_B}
 \hat{\mathit{B}}(u,v)(k)=i\sum_{j\in \mathbb{Z}^2\setminus\{0\}}[(\hat{u}(k-j)\cdot j)\hat{v}(j)-\frac{(\hat{u}(k-j)\cdot j)(\hat{v}(j) \cdot k)}{k \cdot k}k], \forall k\in\mathbb{Z}^2\setminus\{0\}.
\end{align}\\
Well known algebraic properties of $\mathit{B}$(cf. \cite{DFJ05}) include the orthogonality relation 
\begin{align}
\label{orthog_rel}
(\mathit{B}(u,v),w)=-(\mathit{B}(u,w),v), \forall u\in  \mathit{H}, v, w \in \mathit{V},
\end{align}
in particular,
\begin{align}
\label{B_u_v_v}
 (\mathit{B}(u,v),v)=0, \forall u\in  \mathit{H}, v \in \mathit{V},
\end{align}
as well as,
\begin{align}
\label{B_u_u_Au}
(\mathit{B}(u,u),\mathit{A}u)=0, \forall u\in  \mathit{D}({\mathit{A}}).
\end{align}
Also the strong form of enstrophy invariance holds (see, for example \cite{DFJ10}),
\begin{align}
\label{enst_inva}
 (\mathit{B}(\mathit{A}v,v),u)=(\mathit{B}(u,v),\mathit{A}v), \forall u\in  \mathit{H}, v\in \mathit{D}({\mathit{A}}).
\end{align}

\subsection{The Functional Differential Equation}
The functional form of the NSE is
\begin{align}
\label{f_f_NSE}
\frac{du}{dt}+\nu\mathit{A}u+\mathit{B}(u,u)=g,
\end{align}
where $g$ is the Helmholz-Leray projection of $F$ onto the divergence free vector fields. After the rescaling by
\begin{align*}
\tilde{u}=\frac{u}{\nu \kappa_0}, \tilde{t}=\nu \kappa_0^2 t, \tilde{x}=\kappa_0 x, \tilde{\Omega}=[0, 2\pi]^2, \tilde{g}=\frac{g}{\nu^2 \kappa_0^3},
\end{align*}
the corresponding dimensionless functional form is
\begin{align*}
\frac{d\tilde{u}}{d\tilde{t}}+\tilde{A}\tilde{u}+\tilde{B}(\tilde{u},\tilde{u})=\tilde{g},
\end{align*}
where $\tilde{A}$ and $\tilde{B}$ are defined through the rescaled Laplacian and gradient operators. Henceforth, for simplicity, we assume that $\nu=1$, $L=2\pi$ and $\kappa_0=1$. The functional form then becomes
\begin{align}
\label{functional_form}
\frac{du}{dt}+\mathit{A}u+\mathit{B}(u,u)=g
\end{align}
after dropping the tildes. 

Also, recall the following form of Navier-Stokes equation in terms of the Fouier series coefficients,
\begin{align}
\label{fourier_nse}
\frac{d}{dt}\hat{u}(k,t)=\hat{g}(k)-|k|^2\hat{u}(k,t)-\hat{B}(u,u)(k), \text{ for } k\in \mathbb{Z}^2\setminus\{0\},
\end{align}
where $\hat{B}(u,v)(k)$ is given in (\ref{coef_for_B}).

Throughout the paper , we assume that 
\begin{align}
\label{assump_g}
\mathit{A} g= \lambda g, \text{ where } \lambda \in sp(\mathit{A}) \text{ and }\lambda>1.
\end{align}
Since for eigenvector $g$ of Stokes operator $A$, we have $B(g,g)=0$, whose proof can be found in Appendix B of \cite{FJYZ}. This yields that there exists $u_*:=g/\lambda$ that satisfies $Au+B(u,u)=g$, the stationary NSE.  
For $\lambda=1$, see \cite{MARC87}. Then, the energy and enstropy balances are
\begin{align}
\label{en_bal}
\frac{1}{2}\frac{d}{dt}|u|^2=-||u||^2+(g,u),
\end{align}
\begin{align}
\label{ens_bal}
\frac{1}{2}\frac{d}{dt}||u||^2=-|Au|^2+(g,Au)=-|Au|^2+\lambda (g,u).
\end{align}

\subsection{Global Attractor of NSE}
We recall that the global attractor $\mathcal{A}$ of the NSE is the collection of all elements $u_0$ in $H$ for which there exists a solution $u(t)$ of NSE, for all $t\in\mathbb{R}$, such that $u(0)=u_0$ and $\sup_{t\in \mathbb{R}} |u(t)|<\infty$.

To give another definition of $\mathcal{A}$, we need to recall several concepts. First, as is well-known, for any $u_0, g \in H$, there exists a unique continuous function $u$ from $[0,\infty)$ to $H$ such that $u(0)=u_0$, $u(t)\in \mathcal{D}(A)$, $t\in (0, \infty)$, and $u$ satisfies the NSE for all $t\in (0,\infty)$. Therefore, one can define the map $S(t): H\rightarrow H$ by
\begin{equation*}
S(t)u_0=u(t)
\end{equation*}
where $u(\cdot)$ is as above. Since $S(t_1)S(t_2)=S(t_1+t_2)$, the family $\big\{ S(t)\big\}_{t\geq 0}$ is called the ''solution" semigroup. Furthermore, a compact set $\mathcal{B}$ is called absorbing if for any bounded set $\tilde{\mathcal{B}}\subset{H}$ there is a time $\tilde{t}\geq 0$ such that $S(t)\tilde{\mathcal{B}}\subset \mathcal{B}$ for all $t\geq \tilde{t}$. The attractor can be now defined by the formula 
\begin{equation*}
 \mathcal{A}=\bigcap_{t\geq 0} S(t)\mathcal{B}, \label{attractor}
\end{equation*}
where $\mathcal{B}$ is any compact absorbing subset of $H$.

Equivalent characterizations of the global attractor $\mathcal{A}$ are the followings, namely,
\begin{align*}
\mathcal{A}=\text {the smallest set ( in the phase space $H$) which uniformly attracts all compact sets},
\end{align*} and, 

$
\hspace{0.15 in}\mathcal{A}=\text {the largest bounded invariant set in $H$}.
$



\section{Specific preliminaries}

We first recall the definition of ghost solutions first introduced in \cite{CMM13}.
\begin{definition}
A ghost solution is a nonstationary solution $u(\cdot)\in \mathcal{A}$, such that 
\begin{align*}
\dot{e}(t)\equiv \dot{E}(t)\equiv 0, \forall t \in \mathbb{R},
\end{align*}
where $e:=|u(t)|^2$ and $E:=||u(t)||^2$ are referred to as the energy and enstrophy, respectively.
\end{definition}

Using (\ref{en_bal}) and (\ref{ens_bal}), one immediately finds that a ghost solution satisfies
\begin{align}
\label{relation_1}
E=|A^{1/2}u|^2=(g,u)=\frac{1}{\lambda} |Au|^2,
\end{align}\\
hence,
\begin{align}
\label{relation_2}
P:=|Au|^2=(Au,g)=\lambda E.
\end{align}

We also recall the following elementary relations regarding the ghost solutions.

\begin{proposition}
\label{general_relation}
(see also Proposition 6.3 in \cite{CMM13}).\\
 If $u(t)\in \mathcal{A}$ satisfies (\ref{relation_1}), then the following hold,

\begin{align}
\label{Prop1}
\tag{a}
(\dot{u},g)=(\dot{u}, Au)=(\dot{u},u)=0,
\end{align}
\begin{align}
\label{Prop2}
\tag{b}
|B(u,u)|^2+|Au|^2=|\dot{u}|^2+|g|^2,
\end{align}
\begin{align}
\label{Prop3}
\tag{c}
(B(u,u),g)=|B(u,u)|^2-|\dot{u}|^2=|g|^2-|Au|^2,
\end{align}
\begin{align}
\label{Prop4}
\tag{d}
|B(u,u)-g/2|^2=|\dot{u}\pm g/2|^2,
\end{align}
\begin{align}
\label{Prop5}
\tag{e}
\frac{d}{dt}(|B(u,u)|^2)=\frac{d}{dt}(|\dot{u}|^2),
\end{align}
\begin{align}
\label{Prop6}
\tag{f}
(B(u,u),\dot{u})+|\dot{u}|^2=0.
\end{align}
\end{proposition}
\begin{proof}
For the proof, see \cite{CMM13}. We remark here that assuming $u(t)$ satisfies (\ref{relation_1}) is weaker than assuming $u(t)$ is a ghost solution.
\end{proof}

\begin{lemma}
\label{P_and_G^2} 
For any ghost solution $u(t) \in \mathcal{A}$, one has
\begin{align*}
P=|Au(t)|^2<|g|^2=: G^2,
\end{align*}
\end{lemma}

\begin{proof}
It follows from Proposition 3.1(\ref{Prop3}) that 
\begin{align}
\label{inside_2_1}
|B(u,u)|^2-|\dot{u}|^2=G^2-P,
\end{align}
and that 
\begin{align*}
|\dot{u}|^2+(B(u,u),\dot{u})=0,
\end{align*}
which implies 
\begin{align*}
|\dot{u}|^2\leq|B(u,u)||\dot{u}|\leq \frac{|B(u,u)|^2}{2}+\frac{|\dot{u}|^2}{2},
\end{align*}
hence,
\begin{align*}
|\dot{u}|^2\leq |B(u,u)|^2,
\end{align*}
therefore, (\ref{inside_2_1}) gives $P\leq G^2$.

Moreover, the equality occurs only when the equality occurs for the Cauchy-Schwarz inequality; so, if $P=G^2$, then
\begin{align*}
B(u,u)=c\dot{u}, \text{ with } |c|=1,
\end{align*}
and by Proposition 3.1(\ref{Prop6})
\begin{align*}
0&=|\dot{u}|^2+(B(u,u),\dot{u})\\
&=(1+c)|\dot{u}|^2.
\end{align*}

If $c=-1$, then 
\begin{align*}
\dot{u}+B(u,u)=0,
\end{align*}
so by Proposition 3.1(\ref{Prop3}), $Au=g$, and, consequently, $u=A^{-1}g=g/\lambda=u_*$.

If $\dot{u}=0$, then $B(u,u)=0$, so, $Au=g$, which also gives $u=g/\lambda.$

\end{proof}

\begin{lemma}
\label{E_and_e}
The following inequality holds for the energy $e=|u|^2$ and enstrophy $E=|A^{1/2}u|^2$,
\begin{align*}
E\leq \lambda e\leq \lambda E.
\end{align*}
\end{lemma}
\begin{proof}
Indeed, by (\ref{relation_1}) 
\begin{align*}
0&\leq |Au-\lambda u|^2\\
&=|Au|^2-2\lambda|A^{1/2}u|^2+\lambda^2|u|^2\\
&=\lambda(\lambda e-E),
\end{align*}
hence, $\lambda e-E\geq 0$. For the upper bound $\lambda e\leq \lambda E$, recall the Poincar\'e inequality, namely, $|u|\leq |A^{1/2}u|$.
\end{proof}

In the discussion to follow, we will frequently use the sign of the differences of such two quantities, as $P$ and $G^2$, $E$ and $\lambda e$. For convenience, we put them together in the next theorem.

\begin{theorem}
\label{comparison_relation}
The following are equivalent,\\
(i) $P<G^2$,\\
(ii) $E<\lambda e$,\\
(iii) $E^2<eG^2$,\\
(iv) $E^2<eP$.\\
Moreover, if any one of the above inequalities is replaced by equality, then the other three will also become equalities, and the equality occurs only if $u=u_*$.
\end{theorem}

\begin{proof}
(i). From Lemma \ref{P_and_G^2}, we know that $P\leq G^2$, with equality if and only if $u=u_*$. 

(ii). From Lemma \ref{E_and_e}, clearly, $E\leq \lambda e$. 

If $E=\lambda e$, then, from the proof of Lemma \ref{E_and_e}, we see that
\begin{align*}
Au=\lambda u,
\end{align*}
so that $u$ is an eigenvector of the operator $A$, hence $B(u,u)=0$. By taking inner product of the NSE with $u$, using Proposition 3.1(\ref{Prop1}), we get
\begin{align}
\label{u_eigenvector}
(Au,g)=(g,g),
\end{align}
that is, 
\begin{align*}
P=G^2,
\end{align*}
thus, $u=u_*$, as desired.

Conversely, if $u=u_*$, then, since $g=\lambda u_*=\lambda u$, it holds that $E=(u,g)=\lambda(u,u)=\lambda e$.

(iii) and (iv). It follows from the Cauchy-Schwarz inequality that $E=|A^{1/2}u|^2=(u,Au)\leq |u||Au|\leq e^{1/2}P^{1/2}$, hence also $E^2\leq eG^2$, since $P\leq G^2$.

If $E^2=eP$, then by the condition for equality in the Cauchy-Schwarz inequality, $Au=\mu u$, for some $\mu$. Then, combining $P=(Au,g)=\mu(u,g)=\mu E$ with (\ref{relation_1}), we have $\mu=\lambda$. For the same reason as given in case (ii), we have $u=u_*$.

If $E^2=eG^2$, then, since $E^2\leq eP$, we have $G^2\leq P$, hence it must happen that $P=G^2$.

Conversely,  if $u=u_*$, then both $E^2=eG^2$ and $E^2=eP$ will hold.

\end{proof}

The following result gives a lower bound for the enstrophy of the ghost solutions. 
\begin{theorem}
\label{low_bd_ghost}
(see Theorem 6.4 in \cite{CMM13}.)\\
All ghost solutions satisfy
\begin{align}
\label{lower_bd_for_ghost}
E_{*}\geq E\geq \frac{G^2}{\lambda+C_{BG}G(ln(e\lambda))^{1/2}}.
\end{align}
\end{theorem}

Next proposition tells us that any ghost solution $u(t)$ will never be an eigenvector of $A$, in particular, $u(t)$ will not be reduced to be $u_*$. 
\begin{proposition}
\label{prop_1}
(see Proposition 6.2 in \cite{CMM13}).\\
Let $u(t)\in \mathcal{A}\setminus\{u_{*}\}$ be a ghost solution, then 
\begin{align*}
|(A-\mu)u(t)|\geq \max\{\frac{||u_0||}{\lambda^{1/2}},\frac{||u_0||}{G}\},
\end{align*}
for all $t\in \mathbb{R}$, and all $\mu\in sp(A)$ such that $1\leq \mu <\lambda$. And consequently, $u(t)$ is never an eigenvector of $A$.
\end{proposition}

Henceforth, the following relations are always true for the ghost solutions defined in our paper, namely,
\begin{align}
\label{relation_for_constants}
P<G^2,
E<\lambda e,
E^2<eP,
E^2<eG^2.
\end{align}
Notice that these inequalities are strict because the ghost solutions are nonstationary.\\
The parabola bound, namely, $E^2 \leq eG^2$ in our notation, has been obtained in \cite{DFJ05} for all solutions in the global attractor of 2D NSE, regardless of the force.

Actually, 
we can show that if $P$ is close to $G^2$, then $u(t)$ is also close to the stationary point $u_*$.

\begin{proposition}
\label{P_close_G^2}
If $G^2-P=\delta^2$, for some $0<\delta \in \mathbb{R}$, then
\begin{align*}
|A(u-u_*)|=\delta;
\end{align*}
and hence, 
\begin{align*}
|u-u_*|\leq \delta.
\end{align*}
\end{proposition}

\begin{proof}
To see this, notice that (\ref{inside_2_1}) gives,
\begin{align*}
|B(u,u)|^2-|\dot{u}|^2=G^2-P=\delta^2;
\end{align*}
It follows from the NSE, and Proposition 3.1(\ref{Prop6}), that
\begin{align*}
|g-Au|^2&=|\dot{u}+B(u,u)|^2\\
&=|\dot{u}|^2+2(B(u,u),\dot{u})+|B(u,u)|^2\\
&=|B(u,u)|^2-|\dot{u}|^2\\
&=\delta^2.
\end{align*}
Therefore, 
\begin{align*}
|A(u-u_*)|=|Au-Au_*|=|Au-g|=\delta.
\end{align*}
\end{proof}

A similar result follows from the proof of Lemma \ref{E_and_e}. Namely,
\begin{proposition}
If $\lambda e-E=\delta$, for some $0<\delta\in \mathbb{R}$, then 
\begin{align*}
|Au-\lambda u|^2=\lambda \delta.
\end{align*}
\end{proposition}


\section{the reference frame for the NSE}
\subsection {One Choice of the Finite Coordinate System}
\label{subsec_reference}
For a ghost solution $u(\cdot) \in \mathcal{A}$ of NSE, consider the following set of orthonormal system:
\begin{align*}
&f_0:=\frac{g}{|g|}=\frac{g}{G},\\
&f_1:=\frac{u-(u,f_0)f_0}{|u-(u,f_0)f_0|}=\frac{u-\frac{E}{G}f_0}{({e-\frac{E^2}{G^2}})^{1/2}},\\
&f_2:=\frac{\dot{u}}{|\dot{u}|},\\
&f_3:=\frac{Au-(Au,f_0)f_0-(Au,f_1)f_1}{|Au-(Au,f_0)f_0-(Au,f_1)f_1|}=\frac{Au-\frac{P}{G}f_0-\left(\frac{E-\frac{EP}{G^2}}{(e-\frac{E^2}{G^2})^{1/2}}\right)f_1}{\left(P-\frac{P^2}{G^2}-(\frac{E-\frac{EP}{G^2}}{(e-\frac{E^2}{G^2})^{1/2}})^2\right)^{1/2}}.
\end{align*}

We now project each term of the NSE onto the finite dimensional space spanned by $f_0, f_1, f_2$ and $f_3$, called a reference frame. A direct caluculation using Proposition 3.1 gives the following.

\begin{lemma}
\label{coefficient_1}
If we denote by $\tilde{u}$ the column vector containing the projection of $u$ in the $span\{f_0, f_1, f_2, f_3\}$, then
\begin{align*}
\tilde{u}=\begin{pmatrix}
 (u,f_0)\\ 
(u,f_1)\\ 
(u,f_2)\\ 
(u,f_3) 
\end{pmatrix}
=\begin{pmatrix}
\frac{E}{G}\\
(e-\frac{E^2}{G^2})^{1/2}\\
0 \\
0
\end{pmatrix},
\end{align*}
similarly,
\begin{align*}
\widetilde{Au}=\begin{pmatrix}
(Au,f_0)\\
(Au,f_1)\\
(Au,f_2)\\
(Au,f_3)
\end{pmatrix}
=\begin{pmatrix}
\frac{P}{G}\\
\frac{E-\frac{EP}{G^2}}{(e-\frac{E^2}{G^2})^{1/2}}\\
0\\
\left(P-\frac{P^2}{G^2}-(\frac{(E-\frac{EP}{G^2})^2}{e-\frac{E^2}{G^2}})\right)^{1/2}
\end{pmatrix},
\end{align*}

\begin{align*}
\tilde{g}=\begin{pmatrix}
G\\
0\\
0\\
0
\end{pmatrix},
\end{align*}

\begin{align*}
\tilde{\dot{u}}=\begin{pmatrix}
0\\
0\\
|\dot{u}|\\
0
\end{pmatrix},
\end{align*}

and
\begin{align*}
\widetilde{B(u,u)}=\begin{pmatrix}
G-\frac{P}{G}\\
-\frac{E-EP/G^2}{(e-E^2/G^2)^{1/2}}\\
-|\dot{u}| \\
-\left(P-P^2/G^2-\frac{(E-EP/G^2)^2}{e-E^2/G^2}\right)^{1/2}
\end{pmatrix}.
\end{align*}

\end{lemma}

Our first step is to find a matrix, denoted by $\tilde{A}$, which is symmetric and positive definite (SPD), such that 
\begin{align*}
\widetilde{Au}=\tilde{A}\tilde{u}.
\end{align*}
This means that the matrix $\tilde{A}$ behaves similarly as the Stokes operator $A$ in finite dimensions.

To get a desirable SPD matrix, we have to solve a system of equations. This matrix is not unique; one possible choice is the following.
\begin{lemma}
\label{matrix for A}
The following SPD matrix
\begin{align*}
\tilde{A}=\begin{pmatrix}
\lambda & 0 & 0 &0 \\
0 & a & 0 &c \\
0 & 0 &1 & 0 \\
0 & c &0 & b
\end{pmatrix}
\end{align*}
satisfies 
\begin{align*}
\widetilde{Au}=\tilde{A}\tilde{u},
\end{align*}
where 
\begin{align*}
a:=\frac{(Au,f_1)}{(u,f_1)}=\frac{E-EP/G^2}{e-E^2/G^2};
\end{align*}
\begin{align*}
c:=\frac{(Au,f_3)}{(u,f_1)}=\frac{\left(P-P^2/G^2-\frac{(E-EP/G^2)^2}{e-E^2/G^2}\right)^{1/2}}{\left(e-E^2/G^2\right)^{1/2}};
\end{align*}
and $b$ is chosen such that 
\begin{align*}
ab-c^2>0,
\end{align*}
to guarantee the positivity of the matrix $\tilde{A}$.
\end{lemma}

\begin{remark}
We could choose the entry $b$ in the matrix $\tilde{A}$ to be of the form \begin{align*}
b=s\frac{c^2}{a}=s\frac{(Au,f_3)^2}{(Au,f_1)(u,f_1)},
\end{align*}
for any $s \in \mathbb{R}$, with $s>1$. We do not, however, fix the value for this entry, so that we have flexibility to change this if necessary.
\end{remark}

The next step is to find a tensor $\tilde{B}$ to represent the nonlinear term $B(\cdot,\cdot)$ of the NSE in a similar manner; 
namely, we might want it to satisfy 
\begin{align}
\label{tilde_B_1}
\widetilde{B(u,u)}=\tilde{B}(\tilde{u},\tilde{u}),\forall u \in H,
\end{align}
\begin{align}
\label{tilde_B_2}
(\tilde{B}(u,v),w)+(\tilde{B}(u,w),v)=0,\forall u,v,w \in \mathbb{R}^4,
\end{align}
and the analogue of (\ref{enst_inva})
\begin{align}
\label{tilde_B_3}
(\tilde{B}(\tilde{A}v,v),u)-(\tilde{B}(u,v),\tilde{A}v)=0,\forall u,v,w \in \mathbb{R}^4,
\end{align}
notice that the inner product here is taken in $\mathbb{R}^4$.

Recall that the tensor $\tilde{B}$ here is such that 
\begin{align}
\label{B_uv}
\tilde{B}(u,v)=\begin{pmatrix}
\sum_{j,k=0}^{3} B_{jk}^0 u_j v_k \\
\sum_{j,k=0}^{3} B_{jk}^1 u_j v_k \\
\sum_{j,k=0}^{3} B_{jk}^2 u_j v_k \\
\sum_{j,k=0}^{3} B_{jk}^3 u_j v_k 
\end{pmatrix}.
\end{align}
To find all the coefficients $B_{jk}^h$, for $j,k,h=0,\cdots, 3$, we also need to solve a system of equations (\ref{tilde_B_1}), (\ref{tilde_B_2}) and (\ref{tilde_B_3}); however, this system of equations does not possess compatible solutions. So, instead, we look for a tensor $\tilde{B}$ such that the following weaker conditions are met; 
namely,
\begin{equation}
\label{B_1}
\widetilde{B(u,u)}=\tilde{B}(\tilde{u},\tilde{u}),\forall u \in H,
\end{equation}
\begin{equation}
\label{B_2}
0=(\tilde{B}(\tilde{u},v),w)+(\tilde{B}(\tilde{u},w),v),\forall u \in H, \forall v,w \in \mathbb{R}^4,
\end{equation}
\begin{equation}
\label{B_3}
0=(\tilde{B}(\tilde{A}v,v),w)-(\tilde{B}(w,v),\tilde{A}v), \forall v,w \in \mathbb{R}^4.
\end{equation}

To simplify our notation, we denote
\begin{align*}
\tilde{u}&=\begin{pmatrix}
\eta_0 \\
\eta_1\\
\eta_2\\
\eta_3
\end{pmatrix}=\begin{pmatrix}
\eta_0\\
\eta_1\\
0\\
0
\end{pmatrix},
\end{align*}
and 
\begin{align*}
\widetilde{B(u,u)}=\begin{pmatrix}
\beta_0\\
\beta_1\\
\beta_2\\
\beta_3
\end{pmatrix}.
\end{align*}
Writing down these three equations (\ref{B_1}), (\ref{B_2}) and (\ref{B_3}), we get the following equations for the entries $B^h_{jk}$, for $h,j,k=0,\cdots, 3$.
\begin{proposition}
The $B^h_{jk}$'s, for $h,j,k=0,\cdots, 3$, have to satisfy the following constraints,
\begin{align}
\label{B_all_1}
B^h_{00}\eta_0^2+B^h_{01}\eta_0\eta_1+B^h_{10}\eta_0\eta_1+B^h_{11}\eta_1^2=\beta_h, \text{ for all } h=0,\cdots,3,
\end{align}
\begin{align}
\label{B_all_2}
(B^h_{0k}+B^k_{0h})\eta_0+(B^h_{1k}+B^k_{1h})\eta_1=0, \text{ for all } h,k=0,\cdots, 3,
\end{align}
\begin{align}
\label{B_all_3}
B^h_{jk}=B^j_{hk},\text{ for all } h,j,k=0,\cdots,3.
\end{align}
\end{proposition}
\begin{proof}
(\ref{B_all_1}) follows from the definition of tensor (\ref{B_uv}) and the condition (\ref{B_1}). From the definition of inner product in $\mathbb{R}^4$,  (\ref{B_2}) can be written as,
\begin{align*}
0&=(\tilde{B}(\tilde{u},v),w)+(\tilde{B}(\tilde{u},w),v)\\
&=\sum_{h,j,k}(B^h_{jk}\eta_j v_k w_h)+\sum_{h,j,k}(B^h_{jk}\eta_j w_k v_h)\\
&=\sum_{h,j,k}(B^h_{jk}\eta_j v_k w_h+B^k_{jh}\eta_j w_h v_k),
\end{align*}
which gives (\ref{B_all_2}); similarly, (\ref{B_3}) will give (\ref{B_all_3}).

\end{proof}

One solution for the above system of equations (\ref{B_all_1})-(\ref{B_all_3}) is given in the following lemma.

\begin{lemma}
\label{sol_for_B}
If we let the entries for the tensor $\tilde{B}$ be of the form
\begin{align*}
\begin{pmatrix}
B^0_{00}=0 & B^0_{10}=0 & B^0_{01}=\frac{\beta_0}{\eta_0\eta_1} & B^0_{11}=0 \\
B^1_{00}=0 & B^1_{10}=\frac{\beta_1}{\eta_0\eta_1} &B^1_{01}=0 & B^1_{11}=0 \\
B^2_{00}=\frac{\beta_2}{{\eta_0}^2}&B^2_{10}=0 & B^2_{01}=0 &B^2_{11}=0 \\
B^3_{00}=\frac{\beta_3}{{\eta_0}^2} & B^3_{10}=0&B^3_{01}=0 &B^3_{11}=0
\end{pmatrix}
\end{align*}
along with the following relations
\begin{align*}
B^0_{20}=B^2_{00}=-B^0_{02},
\end{align*}
\begin{align*}
B^0_{30}=B^3_{00}=-B^0_{03};
\end{align*}
and the other entries are free, hence we can make them be zeros;  then such $\tilde{B}$ satisfies the relations (\ref{B_all_1})-(\ref{B_all_3}).
\end{lemma}

\subsection{A geometric structure shared by ghost solutions}
We can also describe a geometric structure for the ghost solutions using the reference frame defined in the previous section. 
\begin{theorem}
The reference frame is rigidly carried by the ghost solution.
\end{theorem}

\begin{proof}
For a given ghost solution $u(t)$, with initial value $u(0)=u_0$, without loss of generality, we assume that the vectors $f_0, f_1, f_2, f_3$ are moved to have the same starting point $u_0$.

Notice that $f_0=g/G$ is time independent, so its length and direction will never change in time. As $t$ varies, since $|u(t)|$ and $(u(t),g)=E$ are both time independent, $u(t)$ can only move on a cone. However, $(u,f_1)=(e-\frac{E^2}{G^2})^{1/2}$, $(u,f_2)=0$ and $(u,f_3)=0$, so the angles between $u$ and the $f_j$'s, $j=0,\cdots, 3$ will not change in time and hence the reference frame is rigidly carried by the ghost solution.
\end{proof}

\subsection{Unitary transformation between two ghost solutions}
If there existed two ghost solutions $u(t)$ and $v(t)$, having the same energy $e$ and enstrophy $E$, one may ask: is there any relation between $u(t)$ and $v(t)$? Or, is there an operation that transforms one ghost solution to the other? 

To state our result we need some notation. Let $P(t)$ be the orthogonal projection of the phase space $H$ onto the space $H_{0123}(t):=span\{f_0,f_1(t),f_2(t),f_3(t)\}$, where $f_j,j=0,\cdots, 3$, are given in Section \ref{subsec_reference}. Similarly, if we replace $u(t)$ by $v(t)$ in the definitions of the $f_j$'s, we get another reference frame $\tilde{f}_j,j=0,\cdots, 3$, corresponding to the ghost solution $v(t)$, and let $\tilde{P}(t)$ be the orthogonal projection of the phase space $H$ onto the space $\tilde{H}_{0123}(t)=span\{\tilde{f}_0,\tilde{f}_1(t),\tilde{f}_2(t),\tilde{f}_3(t)\}$. 

Notice that if the linear operator $\zeta\bigotimes \eta: H\rightarrow H$ is defined to be such that $\zeta \bigotimes \eta(h)=\zeta (\eta,h)$, for $\zeta, \eta \in H$, then for any $t\in \mathbb{R}$, $P(t)h=\sum_{j=0}^{j=3}f_j(t)(f_j,h)=\sum_{j=0}^{j=3}f_j(t)\bigotimes f_j(t)(h)$, for $h\in H$, so the operator $P(t)$ can be expressed as $P(t)=\sum_{j=0}^{j=3}f_j(t)\bigotimes f_j(t)$. Similarly, $\tilde{P}(t)=\sum_{j=0}^{j=3}\tilde{f}_j(t)\bigotimes \tilde{f}_j(t)$. Also, for any $t\in \mathbb{R}$, define $W(t): H\rightarrow \tilde{H}_{0123}$ to be $W(t)=\sum_{j=0}^{j=3} \tilde{f}_j\bigotimes f_j$. Under the above notation, we remark that $W(t)P(t)=W(t)$, for all $t\in \mathbb{R}$.

Analyticity of the solutions in the global attractor \cite{Temam83} implies that, there exists $\delta>0$, such that
\begin{align}
\label{condition_for_unitary}
||P(t)-P(t_0)||_{op}<1, \text{ and }  ||\tilde{P}(t)-\tilde{P}(t_0)||_{op}<1,
\end{align}
for all $t_0, t\in \mathbb{R}$ satisfying $|t-t_0|<\delta$, where $||\cdot||_{op}$ is the operator norm in the phase space $H$.

Using the condition (\ref{condition_for_unitary}), we can then borrow the idea of the proof as given in section 110 in the book \cite{RFSN90} to show that there exists a unitary transformation between the reference frames correspoding to different ghost solutions. 

Indeed, for any $t_0 \in \mathbb{R}$ and $t \in (t_0-\delta,t_0+\delta)$, condition (\ref{condition_for_unitary}) guarantees the existence of $T(t)^{-1}$ and $T(t)^{-1/2}$ for the symmetric operator 
\begin{align*}
T(t):=I+(I-P(t_0))(P(t_0)-P(t))(I-P(t_0)),
\end{align*} and also the existence of $\tilde{T}(t)^{-1}$ and $\tilde{T}(t)^{-1/2}$ for the symmetric transformation $\tilde{T}(t)$, where
\begin{equation*} 
\tilde{T}(t):=I+(I-\tilde{P}(t_0))(\tilde{P}(t_0)-\tilde{P}(t))(I-\tilde{P}(t_0)).
\end{equation*} 

One can check that the operator $Z(t)=(I-P(t))T(t)^{-1/2}(I-P(t_0))$ with the adjoint $Z(t)^*=(I-P(t_0))T(t)^{-1/2}(I-P(t))$ satisfies 
\begin{align*}
Z(t)^*Z(t)=I-P(t_0), Z(t)Z(t)^*=I-P(t),
\end{align*} hence the operator 
\begin{align*}
Z(t): (I-P(t_0))H\rightarrow (I-P(t))H,
\end{align*}
is unitary. We could define in a similar way the unitary operator
\begin{equation*}
\tilde{Z}(t): (I-\tilde{P}(t_0))H\rightarrow (I-\tilde{P}(t))H.
\end{equation*}

To make the connection between the two reference frames corresponding to the two ghost solutions $u$ and $v$, we define the operator $W_{\perp}(t):(I-P(t))H\rightarrow (I-\tilde{P}(t))H$ as follows: 
\begin{align*}
W_{\perp}(t)=\tilde{Z}(t) W_{\perp}(t_0) Z(t)^{-1}: (I-P(t))H \overset{Z(t)^{-1}}{\rightarrow}(I-P(t_0))H\overset{W_{\perp}(t_0)}{\rightarrow}(I-\tilde{P}(t_0))H\overset{\tilde{Z}(t)}{\rightarrow}(I-\tilde{P}(t))H,
\end{align*}
where $W_{\perp}(t_0): (I-P(t_0))H\rightarrow (I-\tilde{P}(t_0))H$ is any chosen unitary operator between the two separable infinite dimensional Hilbert spaces $(I-P(t_0))H$ and $(I-\tilde{P}(t_0))H$. Finally, we define the unitary operator $U(t):H\rightarrow H$ to be of the form
\begin{align*}
U(t):=W(t)P(t)+W_{\perp}(t)(I-P(t)),
\end{align*}
which maps $H_{0123}(t)$ unitarily to $\tilde{H}_{0123}(t)$. Therefore, we obtain the following ``uniqueness" result.
\begin{theorem}
\label{uniqueness_for_ghost }
The ghost solutions are unique up to unitary transformations.
\end{theorem}
\begin{remark}
Notice, however, that if we have a sequence of ghost solutions, then they may converge to a stationary solution, the reference frame given above might degenerate to 3 dimensions. In order to avoid this instability, we will introduce a new reference frame in the following section.
\end{remark}

\subsection{A New Reference Frame}

We introduce another set of orthonormal vectors; namely, we define
\begin{align*}
&f_0^{new}:=f_0=\frac{g}{G},\\
&f_1^{new}:=f_1=\frac{u-\frac{E}{G}f_0^{new}}{(e-E^2/G^2)^{1/2}},\\
&f_2^{new}:=f_3
=\frac{Au-\frac{P}{G}f_0^{new}-\left(\frac{E-\frac{EP}{G^2}}{(e-\frac{E^2}{G^2})^{1/2}}\right)f_1^{new}}{\left(P-\frac{P^2}{G^2}-(\frac{E-\frac{EP}{G^2}}{(e-\frac{E^2}{G^2})^{1/2}})^2\right)^{1/2}},\\
&f_3^{new}:=\frac{A^{2}u(t)-\sum_{j=0}^{j=2}(A^{2}u(t),f_j^{new})f_j^{new}}{|A^{2}u(t)-\sum_{j=0}^{j=2}(A^{2}u(t),f_j^{new})f_j^{new}|}.
\end{align*}
We will concentrate on studying this reference frame in the following sections.
\begin{remark}
By similar arguments, we can also show that the new reference frame is rigidly carried by the solution and under this reference frame  there is an unitary transformation between any two ghost solutions.
\end{remark}


\section{a convergence property}
In this section, we study the limiting behavior of a sequence of ghost solutions. Let $GHO$ denote the set of all ghost solutions and denote
\begin{align*}
GHO(e,E)=\{u(\cdot) \in GHO: |u|^2\equiv e, |A^{1/2}u|^2 \equiv E\}.
\end{align*}
Clearly, 
\begin{align*}
GHO=\bigcup_{e, E}GHO(e,E),
\end{align*}
where "many" of the sets $GHO(e,E)=\varnothing$. 


In general, a sequence of ghost solutions can converge to either a fixed point or another ghost. 
Indeed, an easy application of the well known compactness arguments in $C(\mathbb{R},A)$ yields the following. 
\begin{lemma}
\label{conv_lemma}
For any sequence $\{u_n(\cdot)\}_{n=1}^{\infty}\subset GHO$, there exists a subsequence $\{u_{n_k}(\cdot)\}_{k=1}^{\infty}$ which converges in the $|A\cdot|$ norm on any bounded time interval to either

(i) a ghost solution in $GHO(e_{\infty},E_{\infty})$, where 
\begin{align*}
&e_{\infty}=\lim_{k \rightarrow \infty} |u_{n_k}(0)|^2=\lim_{k\rightarrow \infty} e_{n_k},\\
& E_{\infty}=\lim_{k \rightarrow \infty}|A^{1/2}u_{n_k}(0)|^2=\lim_{k\rightarrow \infty} E_{n_k},
\end{align*}
or,

(ii) to a stationary solution $u_{\infty}$ satisfying 
\begin{align*}
&e_{\infty}=|u_{\infty}|^2, \\
&E_{\infty}=|A^{1/2}u_{\infty}|^2.
\end{align*}
Moreover, 

(a) Case (ii) above occurs if and only if there exists $\{t_n\}_{n=1}^{\infty} \in \mathbb{R}$, such that,
\begin{align*}
\lim_{n\rightarrow \infty}|\dot{u}_n(t_n)|=0.
\end{align*}

(b) $P_{\infty}:=\lim_{k}P_{n_k}=\lim_{k\rightarrow \infty} |Au_{n_k}|^2$ exists. And, if $P_{\infty}=E^2_{\infty}/e_{\infty}$, then case (ii) occurs and 
\begin{align*}
u_0=g/\lambda.
\end{align*}

(c) For any given $\alpha>0$, if there exists $\{u_n\}_{n=1}^{\infty} \subset \bigcup_{(\lambda e-E)(E/e)\geq \alpha}GHO(e,E)$, then case (ii) occurs.
\end{lemma}

As a consequence,  the closure the $GHO$ in the space $C(\mathbb{R},A)$ is equal to $GHO \cup \{\text{stationary solutions}\}$.
\begin{remark}
Using the results from the paper \cite{FJLRYZ13}, one could show further that the convergence in case (ii) is also valid for any higher norms.
\end{remark}


\section{the degeneracy problem related to ghost solutions}
Natural problems concern the behavior of the reference frame in Section 4, under both the convergence in Section 5 and as $t\rightarrow \infty$. In order to consider these questions, we introduce for a given ghost solution $u(\cdot)$

\begin{align*}
&H_0=\text{span}\{g\}=\mathbb{R}g,\\
&H_{01}=\text{span}\{g,u\}=\mathbb{R}g+\mathbb{R}u,\\
&H_{012}=\text{span}\{g,u,Au\}=\mathbb{R}g+\mathbb{R}u+\mathbb{R}Au.
\end{align*}
We now look for the conditions that ensure the dimension of $H_{012}$ to be 3. Notice that the two spaces $H_{01}$ and $H_{012}$ both depend on time. However, in the following, we suppress this dependence for simplicity. We start by noticing the following.
\begin{proposition}
\label{degeneracy}
The following is true,
\begin{align*}
H_0 \subseteq H_{01} \subseteq H_{012},
\end{align*}
hence, each space is a proper subspace of the next one.
\end{proposition}

\begin{proof}
If $H_{01}=H_0$, then $u=\mu g$, for some $\mu \in \mathbb{R}$, so $u$ becomes an eigenvector of $A$, it then follows from the NSE that $Au=g$, since $B(u,u)=0$, and $u=u_*=g/\lambda$. Consequently, $Au=\lambda u\in H_{01}$.
If $H_{012}=H_{01}$, say, 
\begin{align}
\label{in_space_deg}
Au=\alpha u+\beta g,
\end{align}
then,  by taking inner product of (\ref{in_space_deg}) with $g$, and $u$, respctively, one has
\begin{align*}
\left\{\begin{matrix}
P=(Au,g)=\alpha E+\beta G^2\\
E=(Au,u)=\alpha e+\beta E.
\end{matrix}\right.
\end{align*}
If 
\begin{align*}
eG^2-E^2\neq 0,
\end{align*}
solving the equations, we get, 
\begin{align}
\label{alpha_beta}
\alpha=\frac{E(G^2-P)}{eG^2-E^2}, \beta=\frac{eP-E^2}{eG^2-E^2}.
\end{align}
Using (\ref{alpha_beta}),  a direct though lengthy computation shows that 
\begin{align*}
|Au-\alpha u-\beta g|^2=0,
\end{align*}
if and only if
\begin{align*}
(eP-E^2)(eG^2-E^2)(G^2-P)=0.
\end{align*}
So, $u=u_*$, by Theorem \ref{comparison_relation}. Therefore, $H_{012}=H_{01}=H_0.$
\end{proof}



Next, let  $P_{012}=P_{012}(t)$ denote the orthogonal projection of $H$ onto the space $H_{012}$. For $\omega \in H$, let $P_{012}\omega=\omega_1 g+\omega_2 u+\omega_3 Au$, then
\begin{align}
\label{matrix_form}
\begin{pmatrix}
(\omega,g)\\
(\omega,u)\\
(\omega, Au)
\end{pmatrix}=M\begin{pmatrix}
\omega_1\\
\omega_2\\
\omega_3
\end{pmatrix},
\end{align}
where, using (\ref{relation_1}), (\ref{relation_2}), we have
\begin{align}
\label{matrix_M}
M=M(e,E)=\begin{pmatrix}
G^2& E &P\\
E &e &E\\
P &E& P
\end{pmatrix},
\end{align}
with determinant
\begin{align*}
det(M)=(\lambda e-E)E(G^2-P)>0.
\end{align*}

If we take $\omega=B(u,u)$, the system (\ref{matrix_form}) becomes
\begin{align*}
\begin{pmatrix}
G^2-P\\
0\\
0
\end{pmatrix}=M\begin{pmatrix}
\omega_1\\
\omega_2\\
\omega_3
\end{pmatrix},
\end{align*}
so 
\begin{align*}
\begin{pmatrix}
\omega_1\\
\omega_2\\
\omega_3
\end{pmatrix}=M^{-1}\begin{pmatrix}
G^2-P\\
0\\
0
\end{pmatrix},
\end{align*} 
that is, 
\begin{align*}
P_{012}B(u(t),u(t))=(g,u(t),Au(t))M^{-1}\begin{pmatrix}
G^2-P\\
0\\
0
\end{pmatrix}.
\end{align*}
By (\ref{matrix_M}), $\omega_1$, $\omega_2$, $\omega_3$ are time independent. Actually, 
\begin{align*}
P_{012}B(u,u)=g-Au,
\end{align*}
and
\begin{align*}
|P_{012}B(u,u)|^2=G^2-P.
\end{align*}
Also, notice that, since
\begin{align*}
\dot{u}(t)+B(u(t),u(t))=g-Au(t)\in H_{012},
\end{align*} we have
\begin{align*}
(1-P_{012})B(u(t),u(t))=-\dot{u}(t),  \forall t\in \mathbb{R}.
\end{align*}

If we take $\omega=A^2u$, then (\ref{matrix_form}) becomes
\begin{align*}
M\begin{pmatrix}
\gamma\\
\beta\\
\alpha\\
\end{pmatrix}=
\begin{pmatrix}
\lambda P\\
P\\
|A^{3/2}u|^2
\end{pmatrix},
\end{align*}
which, using (\ref{matrix_M}), gives,
\begin{align}
\label{eq_for_coef}
\left\{\begin{matrix}
 G^2\gamma+E\beta+P\alpha=\lambda P \\ 
 E\gamma+e\beta+E\alpha=P\\ 
 P\gamma+E\beta+P\alpha=|A^{3/2}u|^2
\end{matrix}\right.
\end{align}
Solving (\ref{eq_for_coef}), one gets
\begin{align*}
&\gamma=\frac{\lambda P-|A^{3/2}u|^2}{G^2-P},\\
&\beta=\frac{\lambda P-|A^{3/2}u|^2}{\lambda e-E},\\ 
&\alpha=\frac{P}{E}-\frac{\lambda P-|A^{3/2}u|^2}{G^2-P}-\frac{e}{E} \frac{\lambda P-|A^{3/2}u|^2}{\lambda e-E}. 
\end{align*}\\
That is,
\begin{align*}
P_{012}A^2u&=(g,u,Au)M^{-1}\begin{pmatrix}
\lambda P\\
P\\
|A^{3/2}u|^2
\end{pmatrix}\\
&=\gamma g+\beta u+\alpha Au.
\end{align*}
Notice that, by (\ref{relation_1})
\begin{align*}
|A^{3/2}u-\lambda A^{1/2}u|^2=|A^{3/2}u|^2-\lambda P \geq 0,
\end{align*}
with equality if and only if
\begin{align*}
A^{3/2}u=\lambda A^{1/2}u,
\end{align*}
which, by (\ref{relation_1}) and Theorem \ref{comparison_relation}(ii), holds if and only if
\begin{align*}
u=u_*=g/\lambda.
\end{align*}
Thus
\begin{align*}
\gamma<0, \beta<0, \text{and } \alpha>\lambda>0;
\end{align*}
and one of the above inequalities becomes an equality if and only if $u=u_*$.

Summarizing the above discussion, we obtain the following proposition.
\begin{proposition}
\label{summerizing_result}
For a ghost solution $u(\cdot)$,
\begin{enumerate}[(i).]
\item If $A^2u(t) \in H_{012}$, for some $t\in \mathbb{R}$, then $A^2u(t)=\gamma g+\beta u(t)+\alpha Au(t)$ with 
\begin{align}
\label{coefs_gamma}
\gamma=\gamma(t)=\frac{\lambda P-|A^{3/2}u(t)|^2}{G^2-P},
\end{align}
\begin{align}
\label{coefs_beta}
\beta=\beta(t)=\frac{\lambda P-|A^{3/2}u(t)|^2}{\lambda e-E},
\end{align}
\begin{align}
\label{coefs_alpha}
\alpha=\alpha(t)=\frac{P}{E}-\frac{\lambda P-|A^{3/2}u(t)|^2}{G^2-P}-\frac{e}{E} \frac{\lambda P-|A^{3/2}u(t)|^2}{\lambda e-E}=\lambda -\gamma(t)-\frac{e}{E} \beta(t).
\end{align}

\item $|A^{3/2}u(t)|^2=\lambda P$ at some $t\in \mathbb{R}$, if and only if $u=u_*$.
\end{enumerate}
\end{proposition}

\begin{remark}
\label{pos_delta}
It is easy to check that $\alpha^2+4\beta>0$ is always true. Indeed, using (\ref{coefs_gamma}), (\ref{coefs_beta}), (\ref{coefs_alpha}), 
\begin{align}
\label{eq_rho}
\alpha^2+4\beta&=(\lambda-\gamma-\frac{e}{E}\beta)^2+4\beta\\ \nonumber
&=(\frac{1}{G^2-P}+\frac{e/E}{2e-E})^2\rho^2\\ \nonumber
&+4\rho(\frac{1-e/E}{2e-E}-\frac{1}{G^2-P})+4,
\end{align}
where 
\begin{align*}
\rho:=\lambda P-|A^{3/2}v|^2<0.
\end{align*}
A simple calculation shows that the discriminant of the quadratic form on the right hand side of (\ref{eq_rho}) is always negative, therefore, $\alpha^2+4\beta$ is always positive.
\end{remark}

In the following, we will discuss the projection of the function $A^2u$ onto the space $H_{012}.$  We define $E_{\mu}$ as the eigenspace projector corresponding to the eigenvalue $\mu \in sp(A)$. 

\begin{theorem}
If the set $C:=\{t \in \mathbb{R}: A^2u(t) \in H_{012}\} \subset \mathbb{R}$ has an accumulation point $t^* \in \mathbb{R}$, then
$\alpha(t),\beta(t),\gamma(t)$ are constants in time.  
\end{theorem}

\begin{proof}
For any $t \in C$, from Proposition \ref{summerizing_result}, we have
$A^2u(t)=\gamma(t) g+\beta(t) u(t)+\alpha(t) Au(t)$. Then for any $\mu\neq \lambda$, $\mu \in sp(A)$, 
\begin{align*}
(\mu^2-\alpha(t) \mu -\beta(t))E_{\mu}u(t)=0.
\end{align*}
By Remark \ref{pos_delta} it follows that that there are two real roots
\begin{align*}
\mu_{\pm}(t)=\frac{\alpha(t) \pm \sqrt{\alpha(t)^2+4 \beta(t)}}{2}.
\end{align*}
Now, analyticity of solutions in global attractor (\cite{Temam83}) implies that $\mu_{\pm}(t)$ are both analytic in $t$. Since the set $C$ is assumed to have an accumulation point, we conclude that $\mu_{\pm}(t)$ are time independent, so $\alpha(t)=\mu_+(t)+\mu_-(t)$, $\beta(t)=-\mu_+(t) \mu_-(t)$ (and hence also $\gamma(t)$), are all time independent.
\end{proof}

\begin{theorem}
If the set $C$ has an accumulation point $t^* \in \mathbb{R}$, then
$ A^2u(t) \in H_{012}$ for all $t \in  \mathbb{R}$.  
\end{theorem}

\begin{proof}
Denote $R(t)=A^2u(t)-\gamma(t) g-\beta(t) u(t)-\alpha(t) Au(t)$, analyticity of solutions in the global attractor (\cite{Temam83}) implies that $R(t)$ is analytic in time. Since $C$ is assumed to have an accumulation point, we conclude that $R(t)$ are time independent, so $R(t)\equiv 0$, for all $t\in \mathbb{R}$.
\end{proof}

\begin{definition}
\label{def_for_cg}
A chained ghost solution is a ghost solution satisfying the following chained relation
\begin{align}
\label{chain_rel}
A^2u(t)=\gamma g+\beta u(t)+\alpha Au(t), \forall t \in \mathbb{R},
\end{align}
for some time independent coefficients $\gamma$, $\beta$, and $\alpha$.
\end{definition}
\begin{remark}
For any given chained ghost solution $u(t)$, if it exists, it follows easily that we also have the chained relation for $A^nu(t)$, $\forall n\in \mathbb{N}_{+}$, that is,
\begin{align*}
A^nu(t)=\gamma_ng+\beta_nu(t)+\alpha_nAu(t), \forall t \in \mathbb{R},
\end{align*}
where $\gamma_n$, $\beta_n$, and $\alpha_n$ are all time independent. 
\end{remark}

From now on, we will focus on chained ghost solutions.

\begin{theorem}
\label{constancy}
The chained ghost solution $u(t)$ has the following decomposition,
\begin{align}
\label{decomp_for_u}
u(t)&=E_{\mu_+}u(t)+E_{\mu_-}u(t)+\eta g\\
&=u_+(t)+u_-(t)+\eta g,\forall t\in \mathbb{R}\nonumber
\end{align}
where $\eta=E/G^2.$

Moreover,
\begin{align}
\label{u_plus}
|u_+|^2=\frac{\lambda-\mu_-}{\mu_+(\mu_+-\mu_-)}E(1-P/G^2),
\end{align}
\begin{align}
\label{u_minus}
|u_-|^2=\frac{\lambda-\mu_+}{\mu_-(\mu_--\mu_+)}E(1-P/G^2).
\end{align}
\end{theorem}

\begin{proof}
If $(1-P_{012})A^{2}u(t)$ vanishes, then
\begin{align*}
A^{2}u=\alpha Au+\beta u+\gamma g,
\end{align*}
where $\alpha$, $\beta$, and $\gamma$ are as given in Definition \ref{def_for_cg}. Equivalently,
\begin{align}
\label{rel_basic}
(A^{2}-\alpha A-\beta) u=\gamma g.
\end{align}

Clearly, (\ref{rel_basic}) implies
\begin{align}
\label{proj_g}
  ({\lambda}^{2}-\alpha {\lambda}-\beta) E_{\lambda} u=\gamma g.
\end{align}
Here we discuss different possibilities.

(i) When $\lambda^2-\alpha \lambda-\beta=0$, whence $\gamma=0$. Then (\ref{eq_for_coef}) becomes 
\begin{align}
\label{simpl_eq_for_coef}
\left\{\begin{matrix}
E\beta+P\alpha=\lambda P \\ 
e\beta+E\alpha=P\\ 
E\beta+P\alpha=|A^{3/2}u|^2.
\end{matrix}\right.
\end{align}
By (\ref{relation_2}), the first and second equations in (\ref{simpl_eq_for_coef}) give
\begin{align*}
(\lambda e-E)\beta=0,
\end{align*}
so that by Theorem \ref{comparison_relation}(ii), we have $\beta=0$;
Then the second and the third equations in (\ref{simpl_eq_for_coef}), together with (\ref{relation_2}), imply
\begin{align*}
\lambda P=|A^{3/2}u|^2,
\end{align*}
which contradicts Proposition \ref{summerizing_result} (ii) . Therefore, $\lambda^2-\alpha \lambda-\beta$ will never vanish.

So, we need to consider the following case only that

(ii) $\lambda^2-\alpha \lambda-\beta$ is never zero. 

From (\ref{proj_g}), 
\begin{align*}
E_{\lambda}u=\frac{\gamma}{\lambda^2-\alpha \lambda-\beta}g;
\end{align*}
and for all $\mu \in sp(A)$, with $\mu \neq \lambda$,
\begin{align}
\label{general_mu}
(\mu^2-\alpha \mu-\beta)E_{\mu}u=0.
\end{align}

By Remark \ref{pos_delta}, $\Delta:=\alpha^2+4\beta>0$, then, (\ref{general_mu}) gives
\begin{align*}
(\mu-\mu_{+})(\mu-\mu_{-})E_{\mu}u=0,
\end{align*}
where $\mu_{\pm}:=\frac{1}{2}(\alpha\pm p)$, with $p=\Delta^{1/2}$. Necessarily, $\lambda \notin \{\mu_{+},\mu_{-}\}$. Also, notice that if $\mu \notin \{\mu_+,\mu_-\}$, then $E_{\mu}u=0$, thus the decomposition of $u$ in terms of the eigenvectors of the operator $A$ only contains three possible nonzero components. Without loss of generality, we may assume that $\mu_{+}>\mu_{-}$. Therefore, in this case, we have,
\begin{align}
\label{decomp_for_u_two}
u=u_++u_-+\eta g,
\end{align}
where $u_+=E_{\mu_{+}}u$ and $u_-=E_{\mu_{-}}u$.

To determine the value of $\eta$, we notice that
\begin{align}
\label{bilinear_two}
B(u,u)=(1-\lambda \eta )g-{\mu}_{+}u_{+}-{\mu}_{-}u_{-}.
\end{align}
Taking the inner product of (\ref{bilinear_two}) with $g$ and using (\ref{relation_1}), (\ref{relation_2}), Proposition \ref{general_relation}(\ref{Prop3}) , we have $\eta= E/G^2$, that is,
\begin{align}
\label{rel_for_u}
u=u_{+}+u_{-}+\frac{E}{G^2}g.
\end{align}
Now, using (\ref{rel_for_u}) and 
\begin{align*}
(B(u,u),u)=0
\end{align*}
gives
\begin{align}
\label{one_for_u}
|A^{1/2}u|^2=\lambda \frac{|A^{1/2}u|^4}{G^2}+\mu_{+}|u_{+}|^2+\mu_{-}|u_{-}|^2;
\end{align}
and,
\begin{align*}
(B(u,u),Au)=0,
\end{align*}
gives
\begin{align}
\label{two_for_u}
\lambda |A^{1/2}u|^2=\lambda^2\frac{|A^{1/2}u|^4}{G^2}+\mu_{+}^2|u_{+}|^2+\mu_{-}^2|u_{-}|^2;
\end{align}
From (\ref{one_for_u}) and (\ref{two_for_u}), we get,
\begin{align}
\label{intermediate_lambda}
\mu_{+}|u_{+}|^2(\lambda-\mu_{+})+\mu_{-}|u_{-}|^2(\lambda-\mu_{-})=0,
\end{align}
which implies
\begin{align}
\label{rel_for_mu's}
\mu_{-}<\lambda<\mu_{+}.
\end{align}
Using (\ref{one_for_u}), (\ref{two_for_u}) to solve for $|u_{-}|^2$ and $|u_{+}|^2$ yields (\ref{u_plus}), (\ref{u_minus}).
\end{proof}
\begin{corollary}
\label{powers_of_A}
For chained ghost solution $u(t)$, we have $|A^su(t)|$ is time independent, for all $s\in \mathbb{R}$.
\end{corollary}
\begin{proof}
For any $s\in \mathbb{R}$, we have,
\begin{align}
\label{power_of_A}
|A^su|^2=\mu_+^{2s}|u_{+}|^2+\mu_-^{2s}|u_{-}|^2+\lambda^{2s}(\frac{E}{G^2})^2|g|^2.
\end{align}
By (\ref{u_plus}), (\ref{u_minus}) the right hand side of (\ref{power_of_A}) is time independent, hence, $|A^{s}u|^2$ is time independent for all $s \in \mathbb{R}$.  Moreover,
\begin{align*}
0&=(A^{2s}u,\frac{du}{dt})\\
&=(A^{2s}u, g-Au-B(u,u))\\
&=(u,A^{2s}g)-|A^{(2s+1)/2}u|^2-(A^{2s}u,B(u,u)),
\end{align*}
that is, $(A^{2s}u, B(u,u))$ is time independent, for all $s\in \mathbb{R}$.
\end{proof}
Now, for $\lambda=2$, by (\ref{rel_for_mu's}), we must have $\mu_{-}=1$, and $\mu_{+}=-\beta$. Hence, solving $A^2u=\gamma g+\beta u+\alpha Au$, using (\ref{rel_for_u}), we have,
\begin{align*}
\gamma g+\beta u+\alpha Au&=\gamma g+\beta(\frac{E}{G^2}g+u_{+}+u_{-})\\
&+ \alpha(\frac{2E}{G^2}g-\beta u_{+}+u_{-})\\
&=A^2u=\frac{4E}{G^2}g+\beta^2 u_{+}+u_{-}.
\end{align*}
from which one gets,
\begin{align}
\label{sum_of_them}
\alpha+\beta=1,
\end{align}
and
\begin{align}
\label{for_constancy}
\frac{4E}{G^2}=\gamma+\beta \frac{E}{G^2}+2\alpha \frac{E}{G^2}.
\end{align}
It follows from the definition of $\alpha$, $\beta$, and $\gamma$,  expressions (\ref{sum_of_them}) and (\ref{for_constancy}), that, 
\begin{align}
\label{for_three_halves}
2P-|A^{3/2}u|^2=\frac{-1}{(1-e/E)(2e-E)^{-1}-(G^2-P)^{-1}}.
\end{align}

After replacing (\ref{for_three_halves}), we get,
\begin{align}
\label{beta_indep}
\beta=\frac{-1}{1-\frac{e}{E}-\frac{2e-E}{G^2-P}}.
\end{align}

Expression (\ref{beta_indep}) tells us where we should look in order to identify a possible chained ghost when $\lambda=2$.
\begin{theorem}
\label{parabola_from_beta}
When $\lambda=2$, the projection onto the $e$-$E$ plane of each chained ghost solution lies on the following parabolic curve, 
\begin{align}
\label{curve_mu}
(2-\mu_+)E^2+(\mu_+-1)G^2E-\mu_+G^2e=0,
\end{align}
for some $\mu_+>2.$
\end{theorem}

\begin{figure}[h!]
  {\includegraphics[width=.5\linewidth]{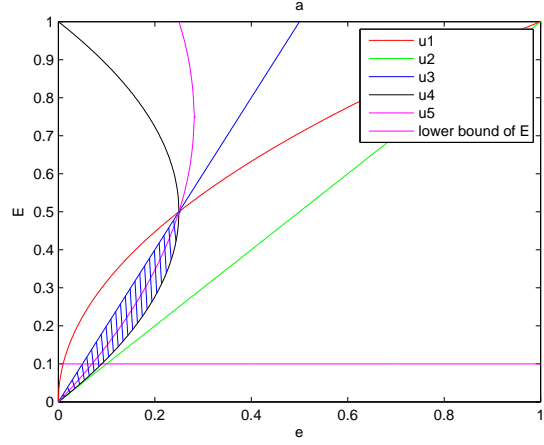}}\hfill
\caption{$\mu_{+}=4$,$u_5: 2E^2-3E+4e=0.$}
\label{figure_one}
\end{figure}

\begin{figure}[h!]
  {\includegraphics[width=.5\linewidth]{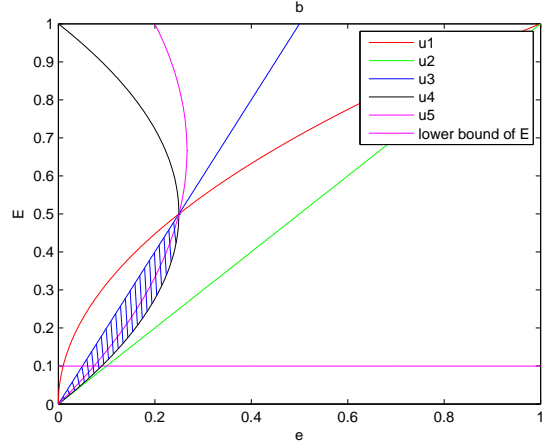}}\hfill
\caption{$\mu_{+}=5$,$u_5: 3E^2-4E+5e=0.$}
\label{figure_two}
\end{figure}

\begin{figure}[h!]
  {\includegraphics[width=.5\linewidth]{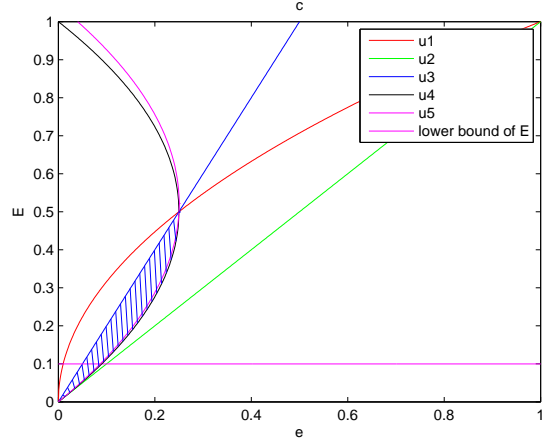}}\hfill
\caption{$\mu_{+}=25$,$u_5: 23E^2-24E+25e=0.$}
\label{figure_three}
\end{figure}

\begin{figure}[h!]
  {\includegraphics[width=0.5\linewidth]{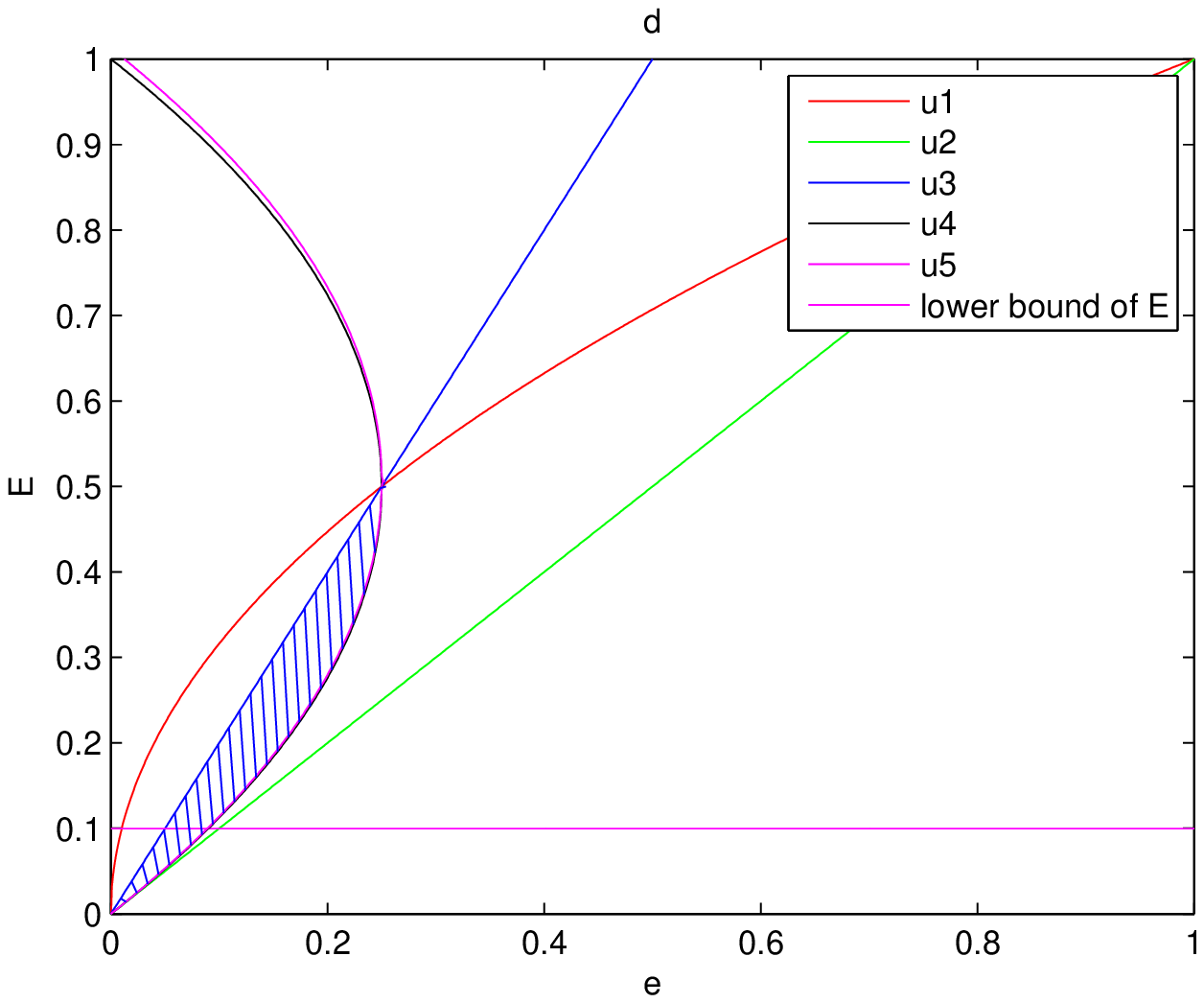}}\hfill
\caption{$\mu_{+}=81$. $u_5: 79E^2-80E+81e=0.$}
\label{figure_four}
\end{figure}

\begin{proof}
First notice that we need $\beta<0$, so from (\ref{beta_indep}), we need, 
\begin{align*}
1-\frac{e}{E}-\frac{2e-E}{G^2-P}>0,
\end{align*}
which is equivalent to 
\begin{align}
\label{parabola_for_ghost}
E^2-G^2E+eG^2<0.
\end{align}
Hence, we have
\begin{align*}
G^4-4eG^2>0,
\end{align*}
equivalently,
\begin{align}
\label{for_e}
e<\frac{G^2}{4};
\end{align}
Also, 
\begin{align*}
\frac{G^2-\sqrt{G^4-4eG^2}}{2}<E<\frac{G^2+\sqrt{G^4-4eG^2}}{2},
\end{align*}
since $E<2e<G^2/2$, we only need to consider
\begin{align}
\label{for_G}
E>\frac{G^2-\sqrt{G^4-4eG^2}}{2}.
\end{align}
Therefore, we got that, if $A^2u(t)\in H_{012}$ for all $t \in \mathbb{R}$, then the projection of the solution $u(t)$ onto the $e$-$E$ plane will be enclosed by the parabola $\{(e,E): E>\frac{G^2-\sqrt{G^4-4eG^2}}{2}\}$, and the classic parabola for the projection of the global attractor, namely, $\{(e,E): E<G \sqrt{e} \}$.\\
Now, to get our result, we use the relation $\mu_+=-\beta$, and
(\ref{beta_indep}), we get
\begin{align}
\label{formula_for_mu_plus}
\mu_+=\frac{1}{1-\frac{e}{E}-\frac{2e-E}{G^2-P}},
\end{align}
which can be rewritten as
\begin{align*}
(2-\mu_+)E^2+(\mu_+-1)G^2E-\mu_+G^2e=0.
\end{align*} for some $\mu_+>2.$

For different values of $\mu_+$, the curve on $e$-$E$ plane defined by (\ref{curve_mu}) will move positions. However, as one could easily check that there are always two intersection points of the curve defined by (\ref{curve_mu}) and the one defined by 
\begin{align}
\label{curve_for_parabola}
E^2-G^2E+eG^2=0,
\end{align}
(see (\ref{parabola_for_ghost})), namely, $(0,0)$ and $(e_*,E_*)=(G^2/4,G^2/2)$; Also, in the limit case when $\mu_+\rightarrow \infty$, the curve (\ref{curve_mu}) becomes (\ref{curve_for_parabola}).

We draw different curves discussed so far in Figure~ {\ref{figure_one}}-Figure~ {\ref{figure_four}} (where we take $G=1$): All graphs have $u_1: E=\sqrt{e}$; $u_2: E=e$; $u_3: E=2e$;
$u_4: E^2-E+e=0$.  In these figures, we also show the lower bound for the enstrophy given by (\ref{lower_bd_for_ghost}) in Theorem \ref{low_bd_ghost}.
\end{proof}


\section{A Galerkin system}

The decomposition (\ref{decomp_for_u_two}) in Theorem \ref{constancy} can be used to exploit the relations between different wavevectors when NSE is written in the form of Fourier modes, as shown in the next theorem.
\begin{theorem}
\label{finite_ode_sys}
Any chained ghost solution $u(t)$ satisfies the following 
Galerkin system,
\begin{align}
\label{fourier_coef_ode}
\frac{d}{dt}\hat{u}(k,t)=\hat{g}(k)-|k|^2\hat{u}(k,t)-[Q_k(\hat{u},\hat{u})-\frac{Q_k(\hat{u},\hat{u})\cdot k}{|k|^2}k],
\end{align}
where 
\begin{align*}
Q_k(\hat{u},\hat{u})=\sum_{\substack{h\in \mathbb{Z}^2\setminus\{0\}\\ |h|^2,|k-h|^2 \in \{\lambda, \mu_+,\mu_-\}}}i(\hat{u}(h,t)\cdot k)\hat{u}(k-h,t), \text{ for } |k|^2 \in\{\lambda, \mu_+,\mu_-\}
\end{align*}
and
\begin{align*}
&\hat{u}(k,t)=0, \text {if } |k|^2\notin \{\lambda, \mu_+,\mu_-\},\\
&\hat{g}(k)=0, \text{if } |k|^2\neq \lambda.
\end{align*}
\end{theorem}

\begin{proof}
It follows immediately from the NSE (\ref{f_f_NSE}),
\begin{align}
\label{big_ode}
\left \{\begin{matrix}
\lambda \eta g+E_{\lambda}B(u,u)=g\\
\frac{d}{dt}u_++\mu_+ u_++E_{\mu_+}B(u,u)=0\\
\frac{d}{dt}u_-+\mu_-u_-+E_{\mu_-}B(u,u)=0
\end{matrix} \right.
\end{align}
Notice that (\ref{big_ode}) is a finite system of differential algebraic equations involving the unknown components for $u_+$ and $u_-$, where $\mu_-<\lambda<\mu_+$.
Clearly, $E_{\mu}B(u,u)=0$, for all $\mu \notin \{\lambda, \mu_+,\mu_-\}$. When (\ref{big_ode}) is written in terms of the Fourier coefficients, invoking (\ref{fourier_nse}) and (\ref{coef_for_B}), we get (\ref{fourier_coef_ode}).\\
\end{proof}
\begin{remark}
Denote $\hat{E}=E_{\mu_+}+E_{\mu_-}+E_{\lambda}$, then we could see from the proof of Theorem \ref{finite_ode_sys} that \begin{align*}
(1-\hat{E})B(u,u)\equiv 0, 
\end{align*}
for  $u$ being a chained ghost solution.
\end{remark}

Theorem \ref{finite_ode_sys} implies the following geometric characterization for chained ghosts. 

\begin{theorem}
\label{ghost_extend_char}
If $u(t)$ is a solution of the Galerkin system in the interval $(0, \bar{t})$ for some $\bar{t}>0$, satisfying 
$(1-\hat{E})B(u(t),u(t))\equiv 0$, for any $t \in (0, \bar{t})$, then $u(t)$ can be extended to be a chained ghost solution.
\end{theorem}

\begin{proof}
Taking the inner product of (\ref{fourier_coef_ode}) with $\hat{u}(t)$ in the Fourier space gives that $|\hat{u}(t)|$ is constant in $t\in (0, \bar{t})$; taking the inner product of (\ref{fourier_coef_ode}) with $A\hat{u}(t)$ in the Fourier space gives $|A^{1/2}\hat{u}(t)|$ is also constant in $t\in (0, \bar{t})$.

By the analyticity of $u(t)$ and $(1-\hat{E})B(u(t),u(t))\equiv 0$, for any $t \in (0, \bar{t})$, we conclude that $u(t)$ can be extended to be a chained ghost solution.
\end{proof}

\begin{remark}
On account of Theorem \ref{finite_ode_sys} and Theorem \ref{ghost_extend_char}, one can use the following computational procedure to check, for given a driving force $g$ and an initial data $u_0$, whether the corresponding solution to the NSE (\ref{f_f_NSE}) is a chained ghost solution. 

First, one solves the Galerkin version ODE obtained by compressing the NSE (\ref{f_f_NSE}) to the space spanned by the eigenvectors corresponding to the eigenvalues $\mu_-$, $\lambda$ and $\mu_+$; to be more precise, one solves the initial value problem for the following ODE over $\hat{E}H$,
\begin{align}
\label{spe_G}
\frac{d}{dt}u(t)+Au(t)+\hat{E}B(u(t),u(t))=\hat{E}g=g,                                                       
\end{align}
with initial value $u(0)=\hat{E}u_0$.

Then, one need check only whether
\begin{align*}
|\eta\prime(t)|<<1,
\end{align*}
and 
\begin{align*}
|A^2u(t)-\gamma g-\beta u(t)-\alpha Au(t)|<<\nu\kappa_0(=1),
\end{align*}
both hold, where $\eta(t):=||u(t)||^2/G^2$, and $\alpha$, $\beta$ and $\gamma$ are as given in Proposition \ref{summerizing_result}.
\end{remark}

Recall that for ghost solutions, we have the decomposition (\ref{decomp_for_u}), namely,
\begin{align*}
u(t)=u_+(t)+u_-(t)+\eta g, \forall t\in \mathbb{R}
\end{align*}
where $u_+ \in E_{\mu_+}$, $u_- \in E_{\mu_-}$ and $g \in E_{\lambda}$, with $\mu_-<\lambda<\mu_+$. We now prove that, using the equation (\ref{spe_G}) (or (\ref{big_ode})), we could identify the value of $\mu_+$, when $\lambda=2$.

\begin{proposition}
\label{mu_plus_value}
Assume $\lambda=2$, then $\mu_+=5$, for any chained ghost $u(t)$, with the decomposition (\ref{decomp_for_u}).
\end{proposition}

\begin{proof}
Taking the dot product of the equation in (\ref{big_ode}) for the component $u_+$ with $u_+$, one gets,
\begin{align*}
\frac{1}{2}\frac{d}{dt}|u_+(t)|^2+\mu_+|u_+|^2&=-(E_{\mu_+}B(u(t),u(t)), u_+(t))\\
&=-(B(u,u),u_+)\\
&=-(B(u_++u_-+\eta g, u_++u_-+\eta g), u_+)\\
&=-\eta (B(u_-, g), u_+)-\eta(B(g,u_-),u_+),
\end{align*}
where, in the last line above, we use that $B(v,v)=0$, for any eigenvector $v$ of the operator $A$ and the relation (\ref{orthog_rel}).

Consider the term $(B(u_,g),u_+)$ and express it in terms of its Fourier coefficients, 
\begin{align}
\label{to_get_mu}
(B(u_-,g),u_+)=\sum_{\substack{h,j,k\in\mathbb{Z}^2\setminus\{0\}\\h+j+k=0}}(\hat{u}_-(h)\cdot j)(\hat{g}(j)\cdot \hat{u}_+(k)).
\end{align}
On the right hand side of expression (\ref{to_get_mu}), for $\hat{u}_-(h)$ not to be zero, we need
\begin{align}
\label{nonzero_u_minus}
h\in S_1:= \{\begin{bmatrix}1\\
0
\end{bmatrix}, \begin{bmatrix} -1\\ 0 \end{bmatrix}, \begin{bmatrix} 0\\1 \end{bmatrix}, \begin{bmatrix} 0\\-1 \end{bmatrix}
\},
\end{align}
and for $\hat{g}(j)$ not to be zero, we need
\begin{align}
\label{nonzero_g}
j\in S_2:=\{\begin{bmatrix}1\\
1
\end{bmatrix}, \begin{bmatrix} -1\\ -1 \end{bmatrix}, \begin{bmatrix} -1\\1 \end{bmatrix}, \begin{bmatrix}1\\-1 \end{bmatrix}
\},
\end{align}
then, if $|k|^2\neq 5$, there is no combination of $h$ from (\ref{nonzero_u_minus}) and $j$ from (\ref{nonzero_g}) satisfying $h+j+k=0$, therefore, each term on the right hand of (\ref{to_get_mu}) will be zero if $|k|^2\neq 5$. Similar arguments work for the Fourier expansion of $(B(g,u_-),u_+)$. Hence, if $\mu_+\neq 5$, then $(B(g,u_-),u_+)\equiv 0$ and $(B(u_-,g),u_+)\equiv 0$, it follows then
\begin{align*}
\frac{1}{2}\frac{d}{dt}|u_+(t)|^2+\mu_+|u_+(t)|^2=0.
\end{align*}
Since $u(t)$ is bounded for all $t\in \mathbb{R}$ we have $|u_+(t)|^2\equiv 0$, which, combined with (\ref{intermediate_lambda}), also implies $|u_-(t)|^2\equiv 0$, so, $u(t)=\eta g$, a steady state. Therefore, the only possible value for $\mu_+$ is 5.
\end{proof}

\section{Nonexistence of chained ghost solutions when $\lambda=2$}
In this section, we will prove that actually, chained ghost solutions do not exist when $\lambda=2$.
The proof relies heavily on the relations (\ref{fourier_coef_ode}). To begin with, by the divergence free convergence condition (\ref{zero_div}) we may set,
\begin{align*}
&\hat{u}(k,t):=i\alpha(k,t)\frac{k^{\perp}}{|k|}\\
&\hat{g}(k):=i\gamma(k)\frac{k^{\perp}}{|k|}\\
&\hat{B}(u,u)(k):=i\beta(u,u)(k)\frac{k^{\perp}}{|k|}=[Q_k(\hat{u},\hat{u})-\frac{Q_k(\hat{u},\hat{u})\cdot k}{|k|^2}k],
\end{align*}
where $\alpha(k,t), \gamma(k)$ and $\beta(u,u)(k,t)$ are scalar functions; and, if $k=\begin{bmatrix} k_1\\k_2 \end{bmatrix}$, then $k^{\perp}=\begin{bmatrix} -k_2 \\k_1 \end{bmatrix}$. 

Notice that the reality condition (\ref{reality_cond}) implies
\begin{align}
\label{reality_for_alpha}
\alpha(-k,t)=\overline{\alpha(k,t)}.
\end{align}

We take the dot product on both sides of (\ref{fourier_coef_ode}) with $-i\frac{k^{\perp}}{|k|}$, using the above notation, to rewrite (\ref{fourier_coef_ode}) as
\begin{align}
\label{fourier_coef_ode2}
\frac{d}{dt}\alpha(k,t)=\gamma(k)-|k|^2\alpha(k,t)+\sum_{\substack{h,j \in \mathbb{Z}^2\setminus\{0\}\\h+j=k\\ |h|^2,|j|^2 \in \{\lambda, \mu_+,\mu_-\}}}\frac{\alpha(h,t)\alpha(j,t)(h^{\perp}\cdot j)(k\cdot j)}{|h||k||j|}, \text{ for } |k|^2 \in \{\lambda, \mu_+,\mu_-\}.
\end{align}\\
Denote the third term on the right hand side of (\ref{fourier_coef_ode2}) as $R_k$. Since $\mu_+=5, \lambda=2, \mu_-=1$, then,
\begin{align}
\label{condition for ka}
&\alpha(k,t)=0, \text{if } |k|^{2} \notin \{1,2,5\}\\
\label{condition for kb}
&\gamma(k)=0, \text{if } |k|^{2} \neq 2
\end{align}
Also, denote
\begin{align*}
S_{3}:=\{\begin{bmatrix}
1\\
2
\end{bmatrix},
\begin{bmatrix}
-1\\
-2
\end{bmatrix},
\begin{bmatrix}
1\\
-2
\end{bmatrix},\begin{bmatrix}
-1\\
2
\end{bmatrix},
\begin{bmatrix}
2\\
1
\end{bmatrix},
\begin{bmatrix}
-2\\
-1
\end{bmatrix},
\begin{bmatrix}
2\\
-1
\end{bmatrix},
\begin{bmatrix}
-2\\
1
\end{bmatrix}\}.
\end{align*}

\begin{theorem}
\label{no_chained}
When $\lambda=2$, there do not exist chained ghost solutions.
\end{theorem}
\begin{proof}
Note that by
(\ref{fourier_coef_ode2}), (\ref{condition for ka}), (\ref{condition for kb}) we have $R_{k}=0$, if $|k|^2 \notin \{1,2,5\}$. Let $k$ be such that $|k|^2 \notin \{1,2,5\}$. In the following, we can get some algebric conditions from this fact.

We will examine all those combinations of the indices $h, j, k$ for which the single term $R_{k}(h,j):=\alpha(h,t) \alpha(j,t) (h^{\perp}\cdot j) (k \cdot j)$ is not zero.
Then, $h$, $j$ and $k$ must be in $S_1 \cup S_2 \cup S_3$. Moreover, the possible $k$'s satisfying $|k|^2 \notin \{1,2,5\}$ and $h+j=k$ are those $k$'s for which $|k|^2 \in\{4,8,9,10,13,18,20\}$.

It is easy to see that if $h$ and $j$ are from the same index set, say $S_l$, $l=1, 2, 3$, then $R_k=0$, so we do not have to consider the possibilities that $|k|^2=18$ and $|k|^2=20$, since then both $h$ and $j$ must be from the set $S_3$. Thus, we are left to consider the cases when $|k|^2 \in\{4,8,9,10,13\}$.

Case 1: $|k|^2=4$; then
\begin{align*}
k\in\{\begin{bmatrix}
2\\
0
\end{bmatrix},
\begin{bmatrix}
-2\\
0
\end{bmatrix},
\begin{bmatrix}
0\\
2
\end{bmatrix},
\begin{bmatrix}
0\\
-2
\end{bmatrix}\}.
\end{align*}
If $k=\begin{bmatrix}
2\\
0
\end{bmatrix}$, then the possible combinations are: 
\begin{enumerate}[(i).]
\item $h=\begin{bmatrix} 0 \\ 1 \end{bmatrix}$, $j=\begin{bmatrix} 2 \\ -1 \end{bmatrix}$, for which $R_{k}(h,j)=-\frac{4}{\sqrt{5}}\alpha (\begin{bmatrix} 0 \\ 1 \end{bmatrix},t) \alpha(\begin{bmatrix} 2 \\ -1 \end{bmatrix}, t)$;

\item $h=\begin{bmatrix} 0 \\ -1 \end{bmatrix}$, $j=\begin{bmatrix} 2 \\ 1 \end{bmatrix}$, for which $R_{k}(h,j)=\frac{4}{\sqrt{5}}\alpha (\begin{bmatrix} 0 \\ -1 \end{bmatrix},t) \alpha(\begin{bmatrix} 2 \\ 1 \end{bmatrix}, t)$; \\
\item $h=\begin{bmatrix} 2 \\ 1 \end{bmatrix}$, $j=\begin{bmatrix} 0 \\ -1 \end{bmatrix}$, for which $R_{k}(h,j)=0$, since $k\cdot j=0$;\\ 
\item $h=\begin{bmatrix} 2 \\ -1 \end{bmatrix}$, $j=\begin{bmatrix} 0 \\1 \end{bmatrix}$, for which $R_{k}(h,j)=0$; \\
\item $h=\begin{bmatrix} 1 \\ 1 \end{bmatrix}$, $j=\begin{bmatrix} 1 \\ -1 \end{bmatrix}$;\\
\item $h=\begin{bmatrix} 1 \\ -1 \end{bmatrix}$, $j=\begin{bmatrix} 1 \\ 1 \end{bmatrix}$; \\
\end{enumerate}
Observe that the contributions of the last two possibilities sum to be zero, since the product $\alpha(h,t) \alpha(j,t)$ does not change when the role of $h$ and $j$ are changed, and $h^{\perp}\cdot j=-j^{\perp}\cdot h$. Therefore, 
\begin{align*}
R_k=\frac{4}{\sqrt{5}}[(\alpha(\begin{bmatrix}
0\\
-1
\end{bmatrix},t))(\alpha(\begin{bmatrix}
2\\
1
\end{bmatrix},t))-(\alpha(\begin{bmatrix}
0\\
1
\end{bmatrix},t))(\alpha(\begin{bmatrix}
2\\
-1
\end{bmatrix},t))]=0,
\end{align*}
or, 
\begin{align}
\label{1}
(\alpha(\begin{bmatrix}
0\\
-1
\end{bmatrix},t))(\alpha(\begin{bmatrix}
2\\
1
\end{bmatrix},t))-(\alpha(\begin{bmatrix}
0\\
1
\end{bmatrix},t))(\alpha(\begin{bmatrix}
2\\
-1
\end{bmatrix},t))=0.
\end{align}
If $k=\begin{bmatrix}
0\\
2
\end{bmatrix}$, we use similar arguments to obtain that

\begin{align*}
R_k=\frac{4}{\sqrt{5}}[(\alpha(\begin{bmatrix}
1\\
0
\end{bmatrix},t))(\alpha(\begin{bmatrix}
-1\\
2
\end{bmatrix},t))-(\alpha(\begin{bmatrix}
-1\\
0
\end{bmatrix},t))(\alpha(\begin{bmatrix}
1\\
2
\end{bmatrix},t))]=0,
\end{align*}
or,  
\begin{align}
\label{2}
(\alpha(\begin{bmatrix}
1\\
0
\end{bmatrix},t))(\alpha(\begin{bmatrix}
-1\\
2
\end{bmatrix},t))-(\alpha(\begin{bmatrix}
-1\\
0
\end{bmatrix},t))(\alpha(\begin{bmatrix}
1\\
2
\end{bmatrix},t))=0.
\end{align}
If $k=\begin{bmatrix}
-2\\
0
\end{bmatrix},$ we have
\begin{align}
\label{3}
(\alpha(\begin{bmatrix}
0\\
1
\end{bmatrix},t))(\alpha(\begin{bmatrix}
-2\\
-1
\end{bmatrix},t))-(\alpha(\begin{bmatrix}
0\\
-1
\end{bmatrix},t))(\alpha(\begin{bmatrix}
-2\\
1
\end{bmatrix},t))=0.
\end{align}
 If $k=\begin{bmatrix}
0\\
-2
\end{bmatrix},$ we have
\begin{align}
\label{4}
(\alpha(\begin{bmatrix}
-1\\
0
\end{bmatrix},t))(\alpha(\begin{bmatrix}
1\\
-2
\end{bmatrix},t))-(\alpha(\begin{bmatrix}
1\\
0
\end{bmatrix},t))(\alpha(\begin{bmatrix}
-1\\
-2
\end{bmatrix},t))=0.
\end{align}

Case 2: $|k|^2=8$; then
\begin{align}
\label{5}
(\alpha(\begin{bmatrix}
1\\
0
\end{bmatrix},t))(\alpha(\begin{bmatrix}
1\\
2
\end{bmatrix},t))-(\alpha(\begin{bmatrix}
0\\
1
\end{bmatrix},t))(\alpha(\begin{bmatrix}
2\\
1
\end{bmatrix},t))=0.
\end{align}
\begin{align}
\label{6}
(\alpha(\begin{bmatrix}
0\\
1
\end{bmatrix},t))(\alpha(\begin{bmatrix}
-2\\
1
\end{bmatrix},t))-(\alpha(\begin{bmatrix}
-1\\
0
\end{bmatrix},t))(\alpha(\begin{bmatrix}
-1\\
2
\end{bmatrix},t))=0.
\end{align}
\begin{align}
\label{7}
(\alpha(\begin{bmatrix}
0\\
-1
\end{bmatrix},t))(\alpha(\begin{bmatrix}
2\\
-1
\end{bmatrix},t))-(\alpha(\begin{bmatrix}
1\\
0
\end{bmatrix},t))(\alpha(\begin{bmatrix}
1\\
-2
\end{bmatrix},t))=0.
\end{align}
\begin{align}
\label{8}
(\alpha(\begin{bmatrix}
-1\\
0
\end{bmatrix},t))(\alpha(\begin{bmatrix}
-1\\
-2
\end{bmatrix},t))-(\alpha(\begin{bmatrix}
0\\
-1
\end{bmatrix},t))(\alpha(\begin{bmatrix}
-2\\
-1
\end{bmatrix},t))=0.
\end{align}

Case 3: $|k|^2=9$; then
\begin{align}
\label{9}
(\alpha(\begin{bmatrix}
2\\
1
\end{bmatrix},t))(\alpha(\begin{bmatrix}
1\\
-1
\end{bmatrix},t))-(\alpha(\begin{bmatrix}
2\\
-1
\end{bmatrix},t))(\alpha(\begin{bmatrix}
1\\
1
\end{bmatrix},t))=0.
\end{align}
\begin{align}
\label{10}
(\alpha(\begin{bmatrix}
-1\\
2
\end{bmatrix},t))(\alpha(\begin{bmatrix}
1\\
1
\end{bmatrix},t))-(\alpha(\begin{bmatrix}
1\\
2
\end{bmatrix},t))(\alpha(\begin{bmatrix}
-1\\
1
\end{bmatrix},t))=0.
\end{align}
\begin{align}
\label{11}
(\alpha(\begin{bmatrix}
-2\\
-1
\end{bmatrix},t))(\alpha(\begin{bmatrix}
-1\\
1
\end{bmatrix},t))-(\alpha(\begin{bmatrix}
-2\\
1
\end{bmatrix},t))(\alpha(\begin{bmatrix}
-1\\
-1
\end{bmatrix},t))=0.
\end{align}
\begin{align}
\label{12}
(\alpha(\begin{bmatrix}
1\\
-2
\end{bmatrix},t))(\alpha(\begin{bmatrix}
-1\\
-1
\end{bmatrix},t))-(\alpha(\begin{bmatrix}
-1\\
-2
\end{bmatrix},t))(\alpha(\begin{bmatrix}
1\\
-1
\end{bmatrix},t))=0.
\end{align}

Case 4: $|k|^2=10$; then similarly,
\begin{align}
\label{13}
(\alpha(\begin{bmatrix}
0\\
1
\end{bmatrix},t))(\alpha(\begin{bmatrix}
1\\
2
\end{bmatrix},t))=0.
\end{align}
\begin{align}
\label{14}
(\alpha(\begin{bmatrix}
1\\
0
\end{bmatrix},t))(\alpha(\begin{bmatrix}
2\\
1
\end{bmatrix},t))=0.
\end{align}
\begin{align}
\label{15}
(\alpha(\begin{bmatrix}
0\\
1
\end{bmatrix},t))(\alpha(\begin{bmatrix}
-1\\
2
\end{bmatrix},t))=0.
\end{align}
\begin{align}
\label{16}
(\alpha(\begin{bmatrix}
-1\\
0
\end{bmatrix},t))(\alpha(\begin{bmatrix}
-2\\
1
\end{bmatrix},t))=0.
\end{align}
\begin{align}
\label{17}
(\alpha(\begin{bmatrix}
0\\
-1
\end{bmatrix},t))(\alpha(\begin{bmatrix}
1\\
-2
\end{bmatrix},t))=0.
\end{align}
\begin{align}
\label{18}
(\alpha(\begin{bmatrix}
0\\
-1
\end{bmatrix},t))(\alpha(\begin{bmatrix}
-1\\
-2
\end{bmatrix},t))=0.
\end{align}
\begin{align}
\label{19}
(\alpha(\begin{bmatrix}
1\\
0
\end{bmatrix},t))(\alpha(\begin{bmatrix}
2\\
-1
\end{bmatrix},t))=0.
\end{align}
\begin{align}
\label{20}
(\alpha(\begin{bmatrix}
-1\\
0
\end{bmatrix},t))(\alpha(\begin{bmatrix}
-2\\
-1
\end{bmatrix},t))=0.
\end{align}

Case 5: $|k|^2=13$; then similarly, we have
\begin{align}
\label{21}
(\alpha(\begin{bmatrix}
1\\
2
\end{bmatrix},t))(\alpha(\begin{bmatrix}
1\\
1
\end{bmatrix},t))=0.
\end{align}
\begin{align}
\label{22}
(\alpha(\begin{bmatrix}
2\\
1
\end{bmatrix},t))(\alpha(\begin{bmatrix}
1\\
1
\end{bmatrix},t))=0.
\end{align}
\begin{align}
\label{23}
(\alpha(\begin{bmatrix}
-1\\
2
\end{bmatrix},t))(\alpha(\begin{bmatrix}
-1\\
1
\end{bmatrix},t))=0.
\end{align}
\begin{align}
\label{24}
(\alpha(\begin{bmatrix}
-2\\
1
\end{bmatrix},t))(\alpha(\begin{bmatrix}
-1\\
1
\end{bmatrix},t))=0.
\end{align}
\begin{align}
\label{25}
(\alpha(\begin{bmatrix}
-1\\
-2
\end{bmatrix},t))(\alpha(\begin{bmatrix}
-1\\
-1
\end{bmatrix},t))=0.
\end{align}
\begin{align}
\label{26}
(\alpha(\begin{bmatrix}
1\\
-2
\end{bmatrix},t))(\alpha(\begin{bmatrix}
1\\
-1
\end{bmatrix},t))=0.
\end{align}
\begin{align}
\label{27}
(\alpha(\begin{bmatrix}
1\\
-1
\end{bmatrix},t))(\alpha(\begin{bmatrix}
2\\
-1
\end{bmatrix},t))=0.
\end{align}
\begin{align}
\label{28}
(\alpha(\begin{bmatrix}
-2\\
-1
\end{bmatrix},t))(\alpha(\begin{bmatrix}
-1\\
-1
\end{bmatrix},t))=0.
\end{align}
Now, we can analyze the system of equations (\ref{1})-(\ref{28}). Notice that the system holds for all $t\in \mathbb{R}$. 
In the following, for notational simplicity, we drop the explicit dependence of time $t$, and write $\alpha(k, t)$ as $\alpha(k)$, for any fixed $t\in \mathbb{R}$. Based on the system of equtions (\ref{1})-(\ref{28}), we can make the following claim:
\begin{claim}
If there is a $k_0\in S_3$ such that $\alpha(k_0)\neq0$, then $\alpha(k)=0$, for all $ k\in S_1\cup S_2$.
\end{claim}
Indeed, one could check all possible choices of $k_0$ to see that the claim is true. Here, we assume,  
say, $\alpha(\begin{bmatrix}
1\\
2
\end{bmatrix})\neq0$, then also $\alpha(\begin{bmatrix}
-1\\
-2
\end{bmatrix})\neq0$, by the reality condition (\ref{reality_for_alpha}).
(\ref{13}) leads to $\alpha(\begin{bmatrix}
0\\
1
\end{bmatrix})=0=\alpha(\begin{bmatrix}
0\\
-1
\end{bmatrix})$. 
From (\ref{25}), we have $\alpha(\begin{bmatrix}
-1\\
-1
\end{bmatrix})=0=\alpha(\begin{bmatrix}
1\\
1
\end{bmatrix})$.
From (\ref{8}), $\alpha(\begin{bmatrix}
-1\\
0
\end{bmatrix})=0=\alpha(\begin{bmatrix}
1\\
0
\end{bmatrix})$.
From (\ref{12}), we see $\alpha(\begin{bmatrix}
1\\
-1
\end{bmatrix})=0=\alpha(\begin{bmatrix}
-1\\
1
\end{bmatrix})$.
Thus $\alpha(k)=0$, for all $ k\in S_1\cup S_2$.

Consequently, either $\alpha(S_1\cup S_2)\equiv0$, or $\alpha(S_3)\equiv 0$, thus, either $u_+\equiv0$, or $u_-\equiv0=g$. Since $g\neq0$, then $u_{+}\equiv0$ , so $|u_+|^2\equiv0$. From (\ref{u_plus}) , it  leads to $P=G^2$, so $u=u^{*}=g/{\lambda}$, which is a contradiction. Therefore, there does not exist chained ghost solutions when $\lambda=2$.
\end{proof}

\section*{acknowledgement}
This work was supported in part by NSF Grant DMS-1109784. The authors would like to thank Professors C. Foias and M. Jolly for suggesting the problem and for subsequent useful discussions and remarks.

\end{document}